\documentclass{IEEEtran}
\usepackage{amsmath}
\usepackage{amssymb}
\usepackage{cite}
\usepackage{algorithmic}
\usepackage{algorithm}
\usepackage{array}
\usepackage{graphicx}
\usepackage{balance}
\usepackage{geometry} 
\usepackage{setspace}
\usepackage{xcolor}

\newtheorem{theorem}{Theorem}

\newtheorem{proof}{Proof}

\begin{document}
	
	\title{Hybrid Satellite-UAV-Terrestrial Networks for 6G Ubiquitous Coverage: A Maritime Communications Perspective}

	\author{Yanmin Wang, 
		Wei~Feng,~\IEEEmembership{Senior Member,~IEEE,} 
		Jue Wang,~\IEEEmembership{Member,~IEEE,}\\ 
		and Tony Q. S. Quek,~\IEEEmembership{Fellow,~IEEE}
		\thanks{Y. Wang is with the China Academy of Electronics and Information Technology, Beijing 100041, China~(email: yanmin-226@163.com).
			
			W. Feng is with the Beijing National Research Center for Information Science and Technology, Department of Electronic Engineering, Tsinghua University, Beijing 100084, China~(email: fengwei@tsinghua.edu.cn).
			
			J. Wang is with the School of Information Science and
			Technology, Nantong University, Nantong 226019, China, also with the
			Peng Cheng Laboratory, Shenzhen 518066, China, and also with the Nantong Research Institute for Advanced Communication Technology, Nantong 226019, China (email: wangjue@ntu.edu.cn).
			
			T. Q. S. Quek is with the Information Systems Technology and Design
			Pillar, Singapore University of Technology and Design, Singapore 487372 (email: tonyquek@sutd.edu.sg).}

	}

	\maketitle
	\begin{abstract}
		In the coming smart ocean era, reliable and efficient communications are crucial for promoting a variety of maritime activities. 
		Current maritime communication networks (MCNs) mainly rely on marine satellites and on-shore base stations (BSs). The former generally provides limited transmission rate, while the latter lacks wide-area coverage capability. Due to these facts, the state-of-the-art MCN falls far behind terrestrial fifth-generation (5G) networks. To fill up the gap in the coming sixth-generation (6G) era, we explore the benefit of deployable BSs for maritime coverage enhancement. Both unmanned aerial vehicles (UAVs) and mobile vessels are used to configure deployable BSs. This leads to a hierarchical satellite-UAV-terrestrial network on the ocean. We address the joint link scheduling and rate adaptation problem for this hybrid network, to minimize the total energy consumption with quality of service (QoS) guarantees. Different from previous studies, we use only the large-scale channel state information (CSI), which is location-dependent and thus can be predicted through the position information of each UAV/vessel based on its specific trajectory/shipping lane. The problem is shown to be an NP-hard mixed integer nonlinear programming problem with
		a group of hidden non-linear equality constraints. We solve it suboptimally by using Min-Max transformation and iterative problem relaxation, leading to a process-oriented joint link scheduling and rate adaptation scheme. As observed by simulations, the scheme can provide agile on-demand coverage for all users with much reduced system overhead and a polynomial computation complexity.
		Moreover, it can achieve a prominent performance close to the optimal solution. 
		
	\end{abstract}
	
	
	\IEEEpeerreviewmaketitle
	
	\section{Introduction}
	With the fast development of blue economy and the construction of smart ocean, maritime communications have attracted ever-increasing research attentions~\cite{p321,p323,p32}. Current maritime communication networks (MCNs) consist of two major component parts, namely the marine satellites (e.g., the Inmarsat) and the on-shore base stations (BSs). While the satellite solution can provide the most wide coverage, it suffers from inherent drawbacks such as the long transmission distance (and thus large delay), as well as limited and expensive bandwidth. On the other hand, due to the geographically limited site locations, the on-shore BSs may only cover a limited offshore area, where coverage holes are inevitable. These facts render the fixed infrastructures, e.g., satellites and on-shore BSs, not efficient and sufficient to fulfill the increasing wide-band-communication demands on the ocean. Actually, the current achieved date rate, not only on the ocean but also in the majority of remote rural areas~\cite{r1}, is way below that of the fifth-generation (5G) cellular network. To enable comparative communication qualities on the ocean with the urban areas, new dynamic network paradigm is desired for wide-band ubiquitous coverage in the coming sixth-generation (6G) era~\cite{r1}.
	
	For a dynamically deployed network, the BSs can be configured on mobile platforms, e.g., unmanned aerial vehicles (UAVs) and vessels, to provide on-demand services on the ocean~\cite{r2, r2_1}. Cooperating with the existing MCN, this leads to a hybrid satellite-UAV-terrestrial network, which has an irregular and dynamic network topology. Resource orchestration, such as link scheduling and rate adaption, now becomes much more complicated in such a hybrid network, considering the uninterrupted backhaul and coordination issues of the deployable BSs. 
	Facing these challenges, we will mathematically formulate a novel design problem, and provide an efficient solution for the hybrid satellite-UAV-terrestrial MCN.
	With the potential to offer enhanced on-demand communications, this hybrid dynamic network paradigm constitutes a competitive alternative to extend the 5G coverage to the oceanic area, and also to enable a ubiquitous coverage for future 6G networks~\cite{r1, r2}.

	\subsection{Related Works}	
    Some recent research efforts have been devoted to characterize the maritime channel and accordingly promote
    the transmission efficiency of on-shore BSs. Wang \textit{et al.}~\cite{r3} developed a comprehensive framework to investigate the near-sea-surface channel. It was observed that the maritime channel is significantly influenced by propagation environments such as sea wave movement and the ducting effect over the sea surface. Particularly, for the impact of sea waves, Huo \textit{et al.}~\cite{r4} investigated the transmission link quality based on the ocean wave modeling for coastal and oceanic waters. 
    As an interesting feature of maritime wireless channels, location-dependent large-scale fading dominates the channel condition. 
    Exploiting this feature, Liu \textit{et al.}~\cite{p404} proposed a hybrid precoding scheme for on-shore BSs, to increase the minimum rate of users, which relies on only large-scale channel state
    information (CSI). Wei \textit{et al.}~\cite{r6} proposed a resource allocation scheme for extending the coverage area of on-shore BSs, which utilizes the shipping lane information.    
    To avoid long-distance transmission for users on a ship, Kim \textit{et al.}~\cite{r5} proposed a hierarchical maritime radio network model, where users communicate with a cluster head equipped on the ship via Wi-Fi, and only the cluster head needs to communicate with on-shore BSs with long-distance transmission technologies. Recently, South Korea has launched a long-term evolution for maritime (LTE-Maritime) research project~\cite{r7}, which targets to cover key offshore areas using on-shore infrastructures. Despite these efforts, the coverage holes of on-shore BSs are inevitable due to their limited height and irregular site locations.   
    
    In addition to on-shore infrastructures, both mobile vessels and UAVs can be utilized to configure deployable BSs for maritime coverage enhancement. Yau \textit{et al.}~\cite{r8} envisioned a picture of multi-hop maritime network connecting various kinds of ships, buoys, and beacons. It was shown that the deployable infrastructure on the ocean is affected by sea surface movement, which may lead to frequent link breakages. 
    In contrast to the vessel-based BS, aerial BSs suffer less from the sea surface movement. For example,
    Teixeira \textit{et al.}~\cite{r9} considered multi-hop airborne communications on the ocean, where the height of tethered flying BSs was optimized to maximize the network capacity. On the other hand, maritime UAVs can 
    also be configured for deployable aerial BS, and they can be utilized for more agile deployment thanks to the high mobility.
    Li \textit{et al.}~\cite{ref_Li} optimized both the trajectory and resource allocation of the UAV, to serve a sailing ship in an accompanying way. Therein the coexistence issue between UAVs and satellites/on-shore BSs have been discussed for the harmonization of the links in the hybrid network. For effective maritime search and rescue, Yang \textit{et al.}~\cite{r10} deployed UAVs in a cognitive mobile computing network, the communication throughput of which was improved 
    with the help of reinforcement learning techniques.
    
    In general, vessels have fixed shipping lanes~\cite{p323}. Therefore, UAV-mounted BSs are more promising for on-demand maritime coverage enhancement. However, current research efforts on UAV communications mainly focus on the terrestrial scenario~\cite{r11,r12}. 
    Among the state-of-the-art research,
    Sun \textit{et al.}~\cite{r1-add} presented an optimal 3D aerial trajectory design and
    wireless resource allocation strategies for solar-powered UAV communication systems.
    Cai \textit{et al.}~\cite{r2-add} focused on the joint trajectory design and resource
    allocation problem for downlink energy-efficient UAV communication systems with eavesdroppers, 
    in which a series of interesting insights were obtained for secure UAV communications.
    While maritime UAVs face some similar challenges to the terrestrial case, e.g., 
    limited on-board energy, which should be elaborately scheduled~\cite{r2-add, r14-2,r13},
    and complicated cooperation pattern~\cite{r16,r17}, they have unique challenges. On the one hand, much more harsh maritime environments require more robust and stronger UAVs, e.g., the oil-powered fixed-wing UAV~\cite{r2}. On the other hand, the interaction between UAVs and other infrastructures becomes a touchy issue. Above the vast ocean, UAVs rely on on-shore BSs or vessel-based BSs for efficient backhaul, and rely on satellites for control signaling. This renders the hybrid network no longer a basic three-node relay model~\cite{r14,r15}. A hierarchical dynamic model is more suitable to characterize it, which however remains unveiled so far.      
    
    In a nutshell, the research on maritime deployable infrastructures is still in its beginning stage, and it is still an open problem to effectively integrate the fixed and deployable infrastructures for 
    maritime coverage enhancement.

	\subsection{Main Contribution}
  We consider a hybrid satellite-UAV-terrestrial MCN with practical constraints. The system consists of a satellite for global signaling coverage over the target area, an on-shore BS, multiple UAVs and multiple vessels. Among all the vessels, some act as deployable BSs and provide communication services for the rest. For the blind hole areas which cannot be reached by either the on-shore BS or the vessel-based BSs, UAVs are dispatched for coverage.
	This hybrid network contains shore-to-vessel (S2V), shore-to-UAV (S2U), vessel-to-UAV (V2U), UAV-to-UAV (U2U), UAV-to-vessel (U2V), and vessel-to-vessel (V2V) links. Our design objective is to minimize the total energy consumption of the entire network by coordinately orchestrating all these links across a certain service period, 
	while guaranteeing the quality of service (QoS) requirements of all users. 
		
	In this hybrid network, all maritime UAVs should keep away from extremely harsh environmental conditions, and most vessels follow fixed shipping lanes for collision prevention. These facts render it impossible to arbitrarily deploy the so-called deployable infrastructures.
	To fully exploit their benefits with these constraints, we should establish close cooperation among all these fixed and deployable infrastructures, making them relying heavily on each other. The hybrid network is thus irregular and indecomposable in general, which is the most important difference between it and existing studies. In this situation, the coupling relationship between the backhaul and fronthaul links of one deployable BS should be carefully satisfied. Otherwise, this BS would turn into a bottleneck of the whole network, leading to wide-band coverage holes. Besides, it becomes challenging to acquire full CSI for resource orchestration, due to limited system overheads, which have already been mostly occupied for handling the irregularity and dynamic coupling of the network.

	The main contributions of this paper are summarized as follows.    
	
	\begin{itemize}
		\item We establish a hierarchical framework for energy-efficient maritime coverage enhancement, by utilizing satellites for global signaling, on-shore BSs for sending source data, UAVs for on-demand filling of the blind spots of coverage, and some vessels for opportunistically relaying data. All these components are orchestrated in a process-oriented manner, by exploiting extrinsic information including the pre-designed trajectories of UAVs and the pre-planned shipping lanes of vessels. This framework may provide an enlightening case for efficiently and agilely integrating space-air-ground-sea fixed and deployable infrastructures. 
		
		\item Under this framework, we mathematically formulate a joint link scheduling and rate adaptation problem, to minimize the energy consumption of the whole hybrid network with QoS guarantees. The optimization problem uses the slowly-varying large-scale CSI only, which can be accurately predicted according to the position information of UAVs and vessels, under the practical composite channel models. The problem is shown to be an NP-hard mixed integer nonlinear programming (MINLP) problem with a group of hidden non-linear equality constraints.
		
		\item We propose an efficient iterative solution to the problem, by leveraging the relaxation and gradually approaching method based on the gentlest ascent principle, as well as the Min-Max transformation. Accordingly, an efficient process-oriented joint link scheduling and rate adaptation scheme is proposed, which has a polynomial computation complexity, and thus is viable for resource-limited practical applications. Moreover, simulation results show that it can achieve a prominent performance close to the optimal solution. This corroborates the effectiveness of process-oriented resource orchestration under the proposed framework. 
	\end{itemize}

   	Note that the optimization of UAV trajectories is not included in our work for the time being, although 
   	it is rather critical~\cite{ref_Li, r1-add, r2-add, r14-2,r13, r16,r17}. According to the divide-and-conquer principle, 
   	we tend to solve the resource orchestration problem first, paving the way for joint trajectory design and resource orchestration. This idea is to some extent necessary, 
   	as the optimization of UAV trajectories will further complicate the resource orchestration problem,
   	which itself is already found NP-hard under the considered practical conditions. Nevertheless, on the basis of some pioneering related works~\cite{r1-add, r2-add, r17},
   	we will try to take the joint optimization problem into account in our future work, 
   	so as to uncover its appealing potential gains in the maritime scenario.
    
	The rest of this paper is organized as follows.
	Section II introduces the system model. In Section III,  
	the joint link scheduling and rate adaptation problem is formulated to minimize the total energy consumption and
	an efficient process-oriented suboptimal solution is proposed.
	In Section IV, simulation results are presented with further discussions. Finally, concluding remarks are given in Section~V.

	\section{System Model} \label{section 2}

\begin{figure*}[!t]
	\centering
	\includegraphics[width=16cm]{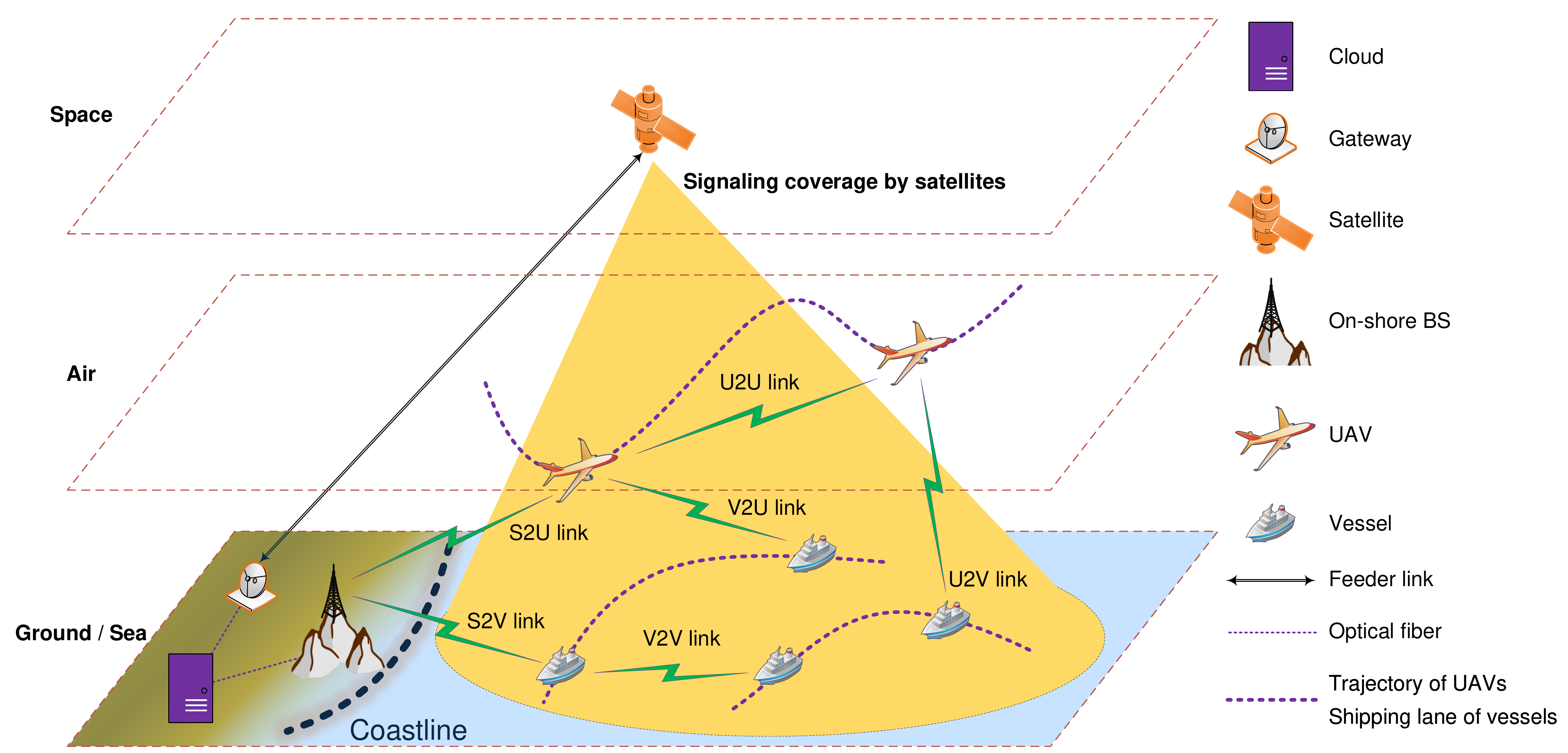}
	\caption{Illustration of a maritime hybrid network with S2V, S2U, V2U, U2U, U2V, and V2V links.}
	\label{fig1}
\end{figure*}	

As shown in Fig.~\ref{fig1}, we consider a hierarchical satellite-UAV-terrestrial MCN. 
The satellite provides signaling for the entire target area.\footnote{Geosynchronous-orbit (GSO) satellites are preferred for the signaling due to their significant superiority in wide-area coverage.
Besides, no harsh real-time signaling is needed for the MCN as the orchestration of the links is supposed to be conducted before the service begins,
which makes the network not that sensitive to the long transmission delay of GSO satellite systems.}
The on-shore BS covers the offshore area, and meanwhile provides backhaul for the UAVs and some of the vessels 
via the S2U links and S2V links, respectively.
The UAVs fly along pre-determined trajectories,\footnote{As illustrated above, the trajectory optimization of UAVs is out of the scope of this paper, which could be jointly considered with resource orchestration in the future work.}
and are utilized for coverage enhancement with the help of selectively established S2U, V2U, U2U, and U2V links.
All the vessels sail along specific shipping-lanes with fixed time tables within the entire service duration~\cite{p323}.
They receive data via the S2V, U2V, or V2V links, and some of them also help relay data for others
through the V2V and V2U links.
All the S2V, S2U, V2U, U2U, U2V, and V2V links are coordinately orchestrated through joint link scheduling and rate adaptation in the cloud, 
based on the large-scale CSI predicted from the pre-determined trajectories of UAVs and the shipping lanes of vessels.
The optimized scheduling and rate adaptation results are then distributed to the UAVs and the vessels for implementation
via the signaling channels provided by the satellite. 

Particularly, we consider one on-shore BS, $I$ UAVs, and $J$ single-antenna vessel users.
The UAVs receive data from the BS and then relay data to the vessels, the first $J'$ vessels with $J' \leq J$ can either
receive data from the BS and UAVs, or transmit data to the UAVs and other vessels.
The other $J-J'$ vessels only receive their own data.
For the ease of exposition, the indices of all transmitters are denoted as $i \in \left\{ {0,1,...,I+J'} \right\}$, and
for the receivers, the indices are $j \in \left\{ {1,...,I+J} \right\}$. 
Specifically, the transmitter with $i = 0$ corresponds to the on-shore BS, $1 \leq i \leq I$ means the transmitter is one of the $I$ UAVs,
and $I+1 \leq i \leq I+J'$ means the transmitter is one of the $J$ vessels that
helps to relay data for other vessels. We assume that there are $N$ ($N \leq I+J$) subcarriers shared by all the links, with subcarrier bandwidth ${B_s}$. 
All the coexisting transmission links will be allocated orthogonal subcarriers, and no co-channel interference exists. This orthogonal transmission assumption 
	enables us to be more focused on the key challenges of the hybrid network.	
	Note that it does not negate this work for scenarios with spectrum reuse,
	where the rise-over-thermal (ROT) parameters could be introduced to model the co-channel interference
    as a composition of the background noise~\cite{ROT-2018, ROT-2015}. 
$P_0$ represents the maximum transmit power of the on-shore BS on a subcarrier.
Similarly, $P_i$ represents the maximum transmit power of transmitter $1\leq i \leq  I + J'$, 
where $1 \leq i \leq I$ indicates that it is for a UAV, 
and $I + 1 \leq i \leq I + J'$ means it is for a vessel.

As shown in Fig.~\ref{fig2}, we divide the total serving time into $T$ time slots. 
Each time slot lasts a duration of $\Delta \tau$. 
For each vessel $j$, we consider a delay-tolerant service with required data volume $V_j^{\rm{QoS}}$. 
Data transmission to user $j$ is required to be completed before $t_j^{\rm{QoS}} \Delta \tau$, as shown in Fig.~\ref{fig2}.
Without loss of generality, 
we assume that different vessels have different desired data, and all the UAVs and vessels operate in a half-duplex mode.
For higher energy efficiency, the links among the BS, the UAVs, and the vessels are scheduled and the link rates 
are jointly adapted on the $N$ subcarries in the $T$ time slots,
with the target of minimizing the total energy consumption of the MCN. 
The joint link scheduling and rate adaptation is carried out based on large-scale CSI of the links, 
which can be predicted based on the pre-determined trajectories of the UAVs as well as 
the specific fixed shipping-lanes and time tables of the vessels.

\begin{figure}
	\centering
	\includegraphics[width=8cm]{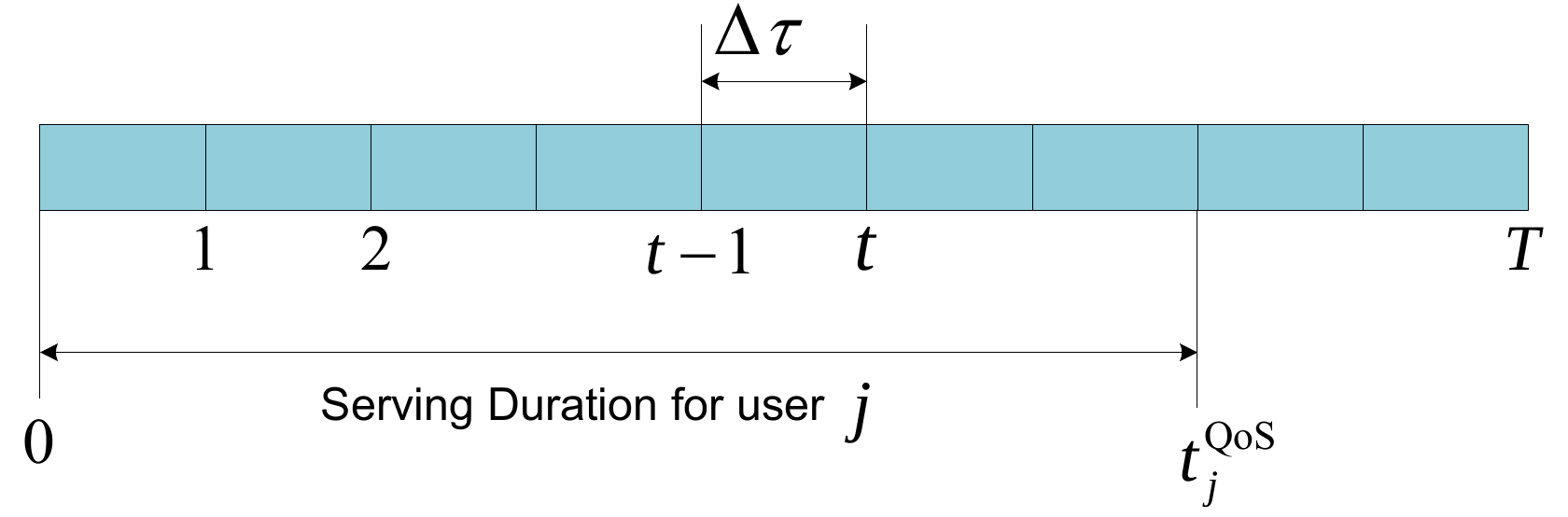}
	\caption{Time slot division and the serving time requirement of the vessels.}
	\label{fig2}
\end{figure}

For a given subcarrier, we denote the
composite channel gain from transmitter $i$ to receiver $j$ at time slot $t $ by $\sqrt {{\beta _{i,j,t }}} {h_{i,j,t }}$.
${h_{i,j,t }}$ denotes the small-scale fading, which  
follows a complex Gaussian distribution, i.e., ${h_{i,j,t }} \sim \mathcal{CN}(0, 1)$. 
It generally varies fast within each time slot.
${\beta _{i,j,t }}$ is the large-scale fading coefficient, which remains constant within each time 
slot.\footnote{Each of the UAVs is supposed to fly at a speed similar to that of the vessels, so that continuous coverage enhancement could be provided for the MCN.
Correspondingly, all ${\beta _{i,j,t }}$ have similar coherence time.}
When $ 1 \leq i \leq I, j=I+1,..., I+J $, or $i=0, I+1,..., I+J',  1 \leq j \leq I$, 
${\beta _{i,j,t }}$ corresponds to a U2V, S2U, or V2U link, respectively.
In this case, as the antenna height at one end of the link (the UAV end) is much higher than the other end, the path loss in dB
can be expressed as~\cite{ref_channel_1}
\begin{equation}\label{eq2_0}
	{\beta _{i,j,t }}[\text{dB}]=\frac{A}{1+ae^{-b(\rho_{i,j,t}-a)}}+B_{i,j,t},
\end{equation}
where
\begin{subequations}
	\begin{align}
		&\!\!\!\!A=\eta_{LOS}-\eta_{NLOS}, \\
		&\!\!\!\!B_{i,j,t}=20\log_{10}(d_{i,j,t})+20\log_{10}\big(\frac{4\pi f_c}{300}\big)+\eta_{NLOS}, \\
		&\!\!\!\!\rho_{i,j,t}=\frac{180}{\pi}\arcsin\big(\frac{h_u}{d_{i,j,t}}\big),
	\end{align}
\end{subequations}
with ${d_{i,j,t }}$ denoting the distance between transmitter $i$ and receiver $j$ at time slot $t $,
$h_u$ being the height of UAV, $f_c$ denoting the carrier frequency in MHz,
and $\eta_{LOS}, \eta_{NLOS}, a, b$ representing environment-related constant parameters. 
For a U2U link, if any exists, the free space path loss is assumed.
For $ i \in \{0, I+1,...,I+J' \} $ and $j \in \{ I+1, ..., I+J \}$, $i \neq j$, 
${\beta _{i,j,t }}$ corresponds to S2V link when $i = 0$, and it corresponds to a V2V link otherwise.
In this case, it is modeled as~\cite{r7}
\begin{equation}{\label{eq_2_1}}
	\begin{split}
		& {\beta _{i,j,t }}[\text{dB}] =(44.9-6.55\log_{10}h_i)\log_{10}\frac{d_{i,j,t}}{1000}+45.5 + \\
		& (35.46-1.1h_j)\log_{10}f_c-13.82\operatorname{log}_{10}h_j +0.7h_j+C,
	\end{split}
\end{equation}
where $h_i$ and $h_j$ denote the antenna heights of the transmitter and the receiver, respectively, 
and $C$ is a constant parameter indicating different propagation environments.
When the location information of the UAVs and vessels is known, 
the large-scale CSI for all the links, i.e., $\beta _{i,j,t }, \forall i,j,t, i \neq j$, can be directly obtained
via these path loss models.

\section{Joint Link Scheduling and Rate Adaptation}
In this section, we first formulate the joint link scheduling and rate adaptation problem for the hybrid MCN. The problem comes out to be an NP-hard MINLP problem. 
To solve it, we introduce
a process-oriented relaxation and gradually approaching method based on the gentlest ascent principle, as well as a Min-Max transformation.
Accordingly, an efficient process-oriented suboptimal joint link scheduling and rate adaptation scheme is proposed, 
with a detailed complexity analysis presented in the end of the section.

\subsection{Problem Formulation }
We denote the link from transmitter $i \in \left\{ {0,1,...,I+J'} \right\}$ to receiver $j \in \left\{ {1,...,I+J} \right\}$ at time slot $t$ by $i \to j@t$.
Let $\delta _{i,j,t} \in \left\{ {0,1} \right\}, i \neq j $, denote 
the scheduling indicator, where
$ {{\delta _{i,j,t}}} = 1$ means the link between transmitter $i$ to receiver $j$ is active at time slot $t$ 
on an allocated subcarrier,
while $ {{\delta _{i,j,t}}} = 0$ means the link is idle.
Note that the subcarrier identifier $n, n=1,...,N$, does not appear in the subscript of $ {{\delta _{i,j,t}}}$.
This is because the large-scale fading is assumed to be the same on different subcarriers
for each link $i \to j@t$, and it is not necessary to distinguish the specific identifier of the subcarrier 
allocated to a link.
Since the total number of subcarriers is $N$, we have\footnote{For simplicity, $i \neq j$ is omitted in (\ref{eq_3_1}) as well as all subsequent expressions.}
\begin{align}{\label{eq_3_1}}
	\sum\limits_{i=0}^{I+J'} {\sum\limits_{j=1}^{I+J} { {{\delta _{i,j,t}}} } } \le N,~~ t \in \{1,...,T\}.
\end{align}
Besides, $\delta_{i,j,t}$ is further constrained by the following inequations  for $\forall t$, since all 
the UAVs and vessels are half-duplex. 
\begin{subequations}  
	\begin{align}{\label{eq_3_2}}
		& \sum\limits_{i=0}^{I+J'} {{{\delta _{i,j,t}}} } + \sum\limits_{j'=1}^{I+J} { {{\delta _{j,j',t}}} \le {1}}, ~j\in \{1,...,I+J'\}, \\
		& \sum\limits_{i=0}^{I+J'} {{{\delta _{i,j,t}}} } \leq 1, ~~~~~~j \in \{I+J'+1,...,I+J\}.
	\end{align}
\end{subequations}

Denote the transmit power for link $i \to j @ t$ by $p_{i,j,t}$, 
where $p_{i,j,t} \leq P_i$ with $P_i$ being the maximum transmit power of transmitter $i$. 
Then the total energy consumption of the MCN can be written as
\begin{align}{\label{eq_3_2_1}}
	{E_{\rm{total}} \left( \{ \delta_{i,j,t} \}, \{p_{i,j,t} \} \right) = \sum\limits_{t = 1}^T { {\sum\limits_{i=0}^{I+J'}{\sum\limits_{j = 1}^{I+J} {{p_{i,j,t}}\delta _{i,j,t}} } \Delta \tau  }}}.
\end{align}
When the large-scale CSI, i.e., $\beta_{i,i,t}, \forall i,j,t$, is available,
the transmission rate of link $i \to j @ t$ can be derived as
\begin{align}{\label{eq_3_3}}
	{r_{i,j,t}} & = {{B_s}\mathbf{E}~{{\log }_2}\left( {1 + \frac{{{p_{i,j,t} }{\beta _{i,j,t }}{{\left| {{h_{i,j,t }}} \right|}^2}}}{{{\sigma ^2}}}} \right)},
\end{align}
where $\mathbf{E}$ denotes the expectation operator with respect to the unknown small-scale fading $h_{i,j,t }$,
and $\sigma ^2$ is the power of additive white Gaussian noise.
Based on the random matrix theory,
${r_{i,j,t}}$ can be accurately approximated by~\cite{JSAC2013}
\begin{align}{\label{eq_3_4}}
	r_{i,j,t} \approx &{B_s} \log_2 \left( 1+ \frac{ \beta _{i,j,t } p_{i,j,t} W_{i,j,t}^{-1} }{ \sigma ^2 }  \right) + B_s \log_2(W_{i,j,t}) \nonumber \\ 
	&- B_s \log_2 (e) \left( 1- W_{i,j,t}^{-1} \right),
\end{align}
with $W_{i,j,t}$ satisfying
\begin{align}{\label{eq_3_5}}
	W_{i,j,t} = 1 + \frac{ \beta _{i,j,t } p_{i,j,t} }{ \sigma ^2 + \beta _{i,j,t } p_{i,j,t} W_{i,j,t}^{-1} }.
\end{align}

In our considered network, the UAVs $i = 1, ..., I$, and vessels $i = I + 1, ..., I + J'$, may 
work as either transmitter or receiver in a certain slot $t$, 
depending on the link scheduling result. 
On the contrary, vessel $i$, $i \in \{I + J' + 1,..., I + J\}$, only receives its own demanded data, and does not transmit to the others. 
Suppose that at the end of the $t$-th slot, UAV or vessel $j$ has a total data volume of $V_{j,t}$.
Then $V_{j,t}$ is given by
\begin{equation}{\label{eq_3_8}}
	V_{j,t} = \left\{
	\begin{aligned}
		& \sum\limits_{i=0}^{I+J'} {{r_{i,j, t }\delta _{i,j,t}}} \Delta \tau,  ~~~t=1, \forall j, \\
		& \sum\limits_{\tau = 1}^t {\left( {\sum\limits_{i=0}^{I+J'} {{r_{i,j,\tau }\delta _{i,j,\tau}}} - \sum\limits_{j'=1}^{I+J} {{r_{j,j',\tau }\delta _{j,j',\tau}}} } \right)
			\Delta \tau }, \\
		&~~~~~~2 \leq t \leq T, 1\leq j \leq I+J', \\
		& \sum\limits_{\tau = 1}^t  {\sum\limits_{i=0}^{I+J'} {{r_{i,j,\tau }\delta _{i,j,\tau}}} }
		\Delta \tau, \\
		&~~~~~~2 \leq t \leq T, I+J'+1 \leq j \leq I+J.
	\end{aligned}
	\right.
\end{equation}
In order to ensure the causality of data forwarding in the network, 
the data that UAV or vessel $j$, $j \in \{ 1,...,I+J'\}$, transmits at time slot $t+1$
should be no more than that it has at the end of time slot $t$, i.e.,
\begin{align}{\label{eq_3_9}}
	&{V_{j,t}} \geq \sum\limits_{j'=1}^{I+J} {{r_{j,j',t+1 }\delta _{j,j',t+1}}} \Delta \tau, \\
	&j \in \{ 1,...,I+J' \}, t \in \{ 0,...,T-1 \}. \nonumber
\end{align}
Note that $V_{j,0}$ refers to the data volume that UAV or vessel $j$ has at the start of time slot $t=1$, and it satisfies
\begin{align}{\label{eq_3_9_0}}
	V_{j,0} = 0, j \in \{ 1,...,I+J \}.
\end{align}
It can be inferred from~(\ref{eq_3_9}) that
\begin{subequations}
\begin{align}{\label{eq_3_9_1}}
\delta_{j, j',1} = 0, j \in \{ 1,...,I+J' \}, j' \in \{ 1,...,I+J \}, \\
r_{j, j',1} = 0, j \in \{ 1,...,I+J' \}, j' \in \{ 1,...,I+J \}.
\end{align}
\end{subequations}
Furthermore, to satisfy the QoS guarantees for the vessels, it is desired that 
\begin{align}{\label{eq_3_10}}
	V_{j,t}|_{t \geq t_j^{\rm{QoS}}} \geq V_j^{\rm{QoS}}, j \in \{ I+1, ...,I+J \}.
\end{align}

Based on (\ref{eq_3_4}),  $p_{i,j,t}$ can be expressed as a function of $r_{i,j,t}$, i.e.,
\begin{align}{\label{eq_3_6_1}}
	p_{i,j,t} = \frac{ \sigma^2 }{ \beta_{i,j,t} } \left( 2^{  \frac{1}{B_s} r_{i,j,t}  + \log_2 (e) ( 1- W_{i,j,t}^{-1} )    } - W_{i,j,t} \right).
\end{align}
Thus, the total energy consumption 
$E_{\rm{total}} \left( \{ \delta_{i,j,t} \}, \{p_{i,j,t} \} \right)$ can be rewritten as
$E_{\rm{total}} \left( \{ \delta_{i,j,t} \}, \{r_{i,j,t} \} \right)$, which is shown in (16).
\newcounter{mytempeqncnt}
\begin{figure*}
	\normalsize
	\setcounter{mytempeqncnt}{\value{equation}}
	\setcounter{equation}{15}
	\begin{equation}{\label{eqn_dbl_x}}
		{E_{\rm{total}} \left( \{ \delta_{i,j,t} \}, \{r_{i,j,t} \} \right) = \sum\limits_{t = 1}^T { {\sum\limits_{i=0}^{I+J'}{\sum\limits_{j = 1}^{I+J} {   \frac{ \sigma^2  }{ \beta_{i,j,t}   } \left( 2^{  \frac{1}{B_s} r_{i,j,t} + \log_2 e ( 1- W_{i,j,t}^{-1} )    } - W_{i,j,t} \right)    } \delta _{i,j,t} \Delta \tau  }   }}}.	
	\end{equation}	
	\setcounter{equation}{\value{mytempeqncnt}}
	\hrulefill
\end{figure*}
Accordingly, the joint link scheduling and rate adaptation problem aiming to minimize 
the network energy consumption can be formulated as
\setcounter{equation}{16}
\begin{subequations}{\label{eq_3_7}}
	\begin{align}
		&\!\!\!\!\!\!\!\!\!\!\!\! \mathop {\min }\limits_{{\left\{ {{\delta _{i,j,t}}} \right\}, \{ r_{i,j,t} \} }}  { E_{\rm{total}} \left( \{ \delta_{i,j,t} \}, \{r_{i,j,t} \} \right) } \\
		{s.t.} \;\; &S^{\delta}_{j,t} \leq 1,~~ 1\leq j \leq I+J, 1\leq t \leq T, \\
		\;\;\;\;\;\; &\sum\limits_{i=0}^{I+J'} {\sum\limits_{j=1}^{I+J} { {{\delta _{i,j,t}}} } } \le N , ~~1\leq t \leq T, \\
		\;\;\;\;\;\; &\left. {{V_{j,t}}} \right|_{t \geq t_j^{{\rm{QoS}}}} \ge V_j^{\rm{QoS}}, ~~I+1 \leq j \leq I+J, \\
		\;\;\;\;\;\; &~ {V_{j,t}} \geq \sum\limits_{j'=1}^{I+J} {{r_{j,j',t+1 }\delta _{j,j',t+1}}} \Delta \tau, \\
		\;\;\;\;\;\; &~1\leq j \leq I+J', 0\leq t \leq T-1,\nonumber \\
		\;\;\;\;\;\; &~0 \leq r_{i,j,t} \leq R_{i,j,t}, \forall i, j, t, \\
		\;\;\;\;\;\; &~\delta _{i,j,t} \in \left\{ {0,1} \right\}, \forall i, j, t, \\
		\;\;\;\;\;\; &~W_{i,j,t} = 1 + \frac{ \beta _{i,j,t } p_{i,j,t} }{ \sigma ^2 + \beta _{i,j,t } p_{i,j,t} W_{i,j,t}^{-1} },\forall i, j, t,
	\end{align}
\end{subequations}
in which 
\begin{equation}{\label{eq_3_7_1}}
	S^{\delta}_{j,t} = \left\{
	\begin{aligned}
		& \sum\limits_{i=0}^{I+J'} {{{\delta _{i,j,t}}} } + \sum\limits_{j'=1}^{I+J} { {{\delta _{j,j',t}}} }, ~~j\in \{1,...,I+J'\}, \\
		& \sum\limits_{i=0}^{I+J'} {{{\delta _{i,j,t}}} }, ~~j \in \{I+J'+1,...,I+J\},
	\end{aligned}
	\right.
\end{equation}
and
\begin{align}{\label{eq_3_3_1}}
{R_{i,j,t}} =  {{B_s}\mathbf{E}~{{\log }_2}\left( {1 + \frac{{{P_i }{\beta _{i,j,t }}{{\left| {{h_{i,j,t }}} \right|}^2}}}{{{\sigma ^2}}}} \right)}.
\end{align}

It can be inferred from~(\ref{eq_3_7}) that with the joint link scheduling and rate adaptation optimization,
	the resultant energy consumption of the UAVs and that of the on-shore BS and the vessels 
	will be mainly decided by the qualities of the links. To tackle the irregularity and dynamic coupling of the network, 
	we target at minimizing the total energy consumption with QoS guarantees.   
	In some extreme cases, this may lead to a violation of the on-board energy constraints of some UAVs~\cite{r2-add, r14-2,r13}.
	To eliminate this risk, 
	one somewhat coarse but effective approach  
	is to adjust the maximum number of active links where the UAVs with limited on-board energy
	act as transmitters. Concretely,   
	when there is a limitation, i.e., $\bar{E}_{\bar{j}}$, on the energy consumption for UAV $\bar{j}, 1 \leq \bar{j} \leq I$, within the service period,
	an extra constraint $\sum_{t=1}^{T} \sum_{j'=1}^{I+J} \delta_{\bar{j}, j', t} \leq \bar{E}_{\bar{j}} / P_{\bar{j}} $ 
	could be added to~(\ref{eq_3_7}).
	In the following, it will be seen that the extra constraint makes no difference to the derivations of the proposed scheme.
	Nevertheless, more elaborate approaches for incorporating the energy limitation of UAVs could be further explored in the future work.

\subsection{Problem Analysis}
To solve the problem in~(\ref{eq_3_7}), we are confronted with two challenges. 
First, (\ref{eq_3_7}) is a MINLP problem, 
which is NP-hard according to~\cite{NP}. 
Second, the closed-form approximation for $r_{i,j,t}$ in~(\ref{eq_3_4}) 
brings a group of hidden non-linear equality constraints on the auxiliary variables $W_{i,j,t}$ and $r_{i,j,t}$,
as shown in~(\ref{eq_3_7}h). 

In the folllowing, we analyze (\ref{eq_3_7}) in allusion to these two challenges, 
and try to pave the way to a suboptimal but efficient solution.
Specifically, the relationship between $r_{i,j,t}$ and $\delta_{i,j,t}$
and that between $r_{i,j,t}$ and $W_{i,j,t}$ are analyzed and utilized.
According to the analysis, we show
that the integer scheduling indicators $\{ \delta_{i,j,t} \}$ are
able to be effectively processed with a merging and detaching strategy 
based on a process-oriented relaxation and gradually-approaching method,
and the non-linear equality constraints on $W_{i,j,t}$ can be delicately circumvented 
through a Min-Max transformation of the problem.

\emph{1) Relaxation and gradually approaching}

Note that $r_{i,j,t}$ and $\delta_{i,j,t}$ are 
closely related in the solutions for (\ref{eq_3_7}). For $\forall i,j,t$, the relation is described as
\begin{subequations}{\label{eq_3_11}}
	\begin{align}
		r_{i,j,t} = 0 \leftrightarrow \delta_{i,j,t} = 0 , \\
		r_{i,j,t} > 0 \leftrightarrow \delta_{i,j,t} = 1.
	\end{align} 
\end{subequations}
It means that if the scheduling indicator $\delta_{i,j,t}$ is $0$/$1$, i.e., the link $i \to j@t$ is set to be inactive/active, 
then the transmission rate on that link must be allocated a zero/non-zero value, and vise versa.
Based on this observation, $\delta_{i,j,t}$ can be relaxed by being merged into $r_{i,j,t}$.
The new problem after relaxation can be expressed as
\begin{subequations}{\label{eq_3_14}}
	\begin{align}
		&\!\!\!\!\!\!\!\!\!\!\!\! \mathop {\min }\limits_{{ \{ r_{i,j,t} \} }}  { \tilde{E}_{\rm{total}} \left( \{r_{i,j,t} \} \right) } \\
		{s.t.} \;\; & S^r_{j,t} \le 1, ~~1 \leq j\leq I+J, 1\leq t \leq T, \\
		\;\;\;\;\;\; &\sum\limits_{i=0}^{I+J'} {\sum\limits_{j=1}^{I+J} { \frac{{r_{i,j,t}}}{ R_{i,j,t} } } } \le N, ~~ 1 \leq t \leq T, \\
		\;\;\;\;\;\; &\left. {{\tilde{V}_{j,t}}} \right|_{t \geq t_j^{{\rm{QoS}}}} \ge V_j^{\rm{QoS}}, ~~I+1 \leq j \leq I+J, \\
		\;\;\;\;\;\; &~{\tilde{V}_{j,t}} \geq \sum\limits_{j'=1}^{I+J} {{r_{j,j',t+1 }}} \Delta \tau, \\
		\;\;\;\;\;\; &~1\leq j \leq I+J',   0\leq t \leq T-1, \nonumber\\
		\;\;\;\;\;\; &~0 \leq r_{i,j,t} \leq R_{i,j,t}, \forall i, j, t, \\
		\;\;\;\;\;\; &~W_{i,j,t} = 1 + \frac{ \beta _{i,j,t } p_{i,j,t} }{ \sigma ^2 + \beta _{i,j,t } p_{i,j,t} W_{i,j,t}^{-1} },\forall i, j, t,
	\end{align}
\end{subequations}
where $\tilde{E}_{\rm{total}} \left( \{r_{i,j,t} \} \right)$ is shown in (22),
\begin{figure*}
	\normalsize
	\setcounter{mytempeqncnt}{\value{equation}}
	\setcounter{equation}{21}
	\begin{equation}{\label{eqn_dbl_x}}
		{\tilde{E}_{\rm{total}} \left( \{r_{i,j,t} \} \right) = \sum\limits_{t = 1}^T { {\sum\limits_{i=0}^{I+J'}{\sum\limits_{j = 1}^{I+J} {   \frac{ \sigma^2  }{ \beta_{i,j,t}   } \left( 2^{  \frac{1}{B_s} r_{i,j,t} + \log_2 e ( 1- W_{i,j,t}^{-1} )    } -W_{i,j,t} \right)    }  \Delta \tau  }   }}}.	
	\end{equation}	
	\setcounter{equation}{\value{mytempeqncnt}}
	\hrulefill
\end{figure*}
\setcounter{equation}{22}
and 
\begin{equation}{\label{eq_3_14_1}}
	S^{r}_{j,t} = \left\{
	\begin{aligned}
		& \sum\limits_{i=0}^{I+J'} { \frac{r _{i,j,t}}{ R_{i,j,t} } } + \sum\limits_{j'=1}^{I+J} { \frac{r _{j,j',t}}{ R_{j,j',t} } }, ~~j\in \{1,...,I+J'\}, \\
		& \sum\limits_{i=0}^{I+J'} { \frac{r _{i,j,t}}{ R_{i,j,t} } }, ~~j \in \{I+J'+1,...,I+J\},
	\end{aligned}
	\right.
\end{equation}
\begin{equation}{\label{eq_3_15}}
	\tilde{V}_{j,t} = \left\{
	\begin{aligned}
		& \sum\limits_{i=0}^{I+J'} {{r_{i,j, t }}} \Delta \tau,  ~~~t=1, \forall j, \\
		& \sum\limits_{\tau = 1}^t {\left( {\sum\limits_{i=0}^{I+J'} {{r_{i,j,\tau }}} - \sum\limits_{j'=1}^{I+J} {{r_{j,j',\tau }}} } \right)
			\Delta \tau }, \\
		&~~~~~~2 \leq t \leq T, 1\leq j \leq I+J', \\
		& \sum\limits_{\tau = 1}^t  {\sum\limits_{i=0}^{I+J'} {{r_{i,j,\tau }}} }
		\Delta \tau , \\
		&~~~~~~2 \leq t \leq T, I+J'+1 \leq j \leq I+J.
	\end{aligned}
	\right.
\end{equation}
Note that as compared to the original problem (\ref{eq_3_7}), the constraints
(\ref{eq_3_7}b) and (\ref{eq_3_7}c) are now transformed into (\ref{eq_3_14}b)  and (\ref{eq_3_14}c), respectively, 
where the original constraints on 
$\delta_{i,j,t}$ is now implicitly expressed in terms of $r_{i,j,t}$ via the relations described in~(\ref{eq_3_11}).

Denote an optimal solution for (\ref{eq_3_14}) as $\{ \tilde{r}_{i,j,t} \}$,
and then the corresponding scheduling indicators, i.e.,  $\{ \tilde{\delta}_{i,j,t} \}$, 
can be detached from $\{ \tilde{r}_{i,j,t} \}$ according to (\ref{eq_3_11}).
Because of relaxation, 
the obtained $\{\tilde{\delta}_{i,j,t}\}$ might be not able to satisfy the original constraints described in 
(\ref{eq_3_7}b) and (\ref{eq_3_7}c),
as $\{ \tilde{r}_{i,j,t} \}$ and $\{ \tilde{\delta}_{i,j,t} \}$ are found from a larger solution space.  
However, it can be inferred that $\{ \tilde{\delta}_{i,j,t} \}$ indicates the links with \emph{the best qualities} in the network.
Thus, we may approach a solution for the original problem in~(\ref{eq_3_7})
based on $\{ \tilde{r}_{i,j,t} \}$ and $\{ \tilde{\delta}_{i,j,t} \}$, 
through gradually lessening the violations.
This can be achieved by shrinking the solution space of the problem~(\ref{eq_3_14}) gradually,
referring to that of~(\ref{eq_3_7}), with the help of $\{ \tilde{r}_{i,j,t} \}$ and $\{ \tilde{\delta}_{i,j,t} \}$.

To this end, we develop a process-oriented relaxation and gradually-approaching method
for the original problem~(\ref{eq_3_7}) in the following. 
Specifically, we first check the constraints in (\ref{eq_3_7}b) and (\ref{eq_3_7}c) with $\{ \tilde{\delta}_{i,j,t} \}$ and find those that are violated.
Then, a smaller solution space is derived for~(\ref{eq_3_14})  based on the gentlest-ascent principle,
in which the total energy consumption of the network always ascends the slowest.
Lastly, we
update $\{ \tilde{r}_{i,j,t} \}$ and $\{ \tilde{\delta}_{i,j,t} \}$ within the obtained smaller solution space,
aiming to lessen the violations of the constraints in (\ref{eq_3_7}b) and (\ref{eq_3_7}c).
The above operations are carried out iteratively, 
until all the constraints in (\ref{eq_3_7}b) and (\ref{eq_3_7}c) are satisfied.

Firstly, based on $\{ \tilde{\delta}_{i,j,t} \}$, we determine the violated constraints in (\ref{eq_3_7}b) and (\ref{eq_3_7}c),
and find all non-zero $\tilde{\delta}_{i,j,t}$ involved in the violated constraints.
For $t=1,...,T$ and $x=1,2$, define
\begin{align}{\label{eq_3_50}}
	\mathfrak{\Delta}_{t,x} = \{  (i^{(t,x)}_m, j^{(t,x)}_m, t) | m=1, ..., M_{t,x}  \}, 
\end{align}
where $(i^{(t,x)}_m, j^{(t,x)}_m, t)$ satisfies $\tilde{\delta}_{i^{(t,x)}_m, j^{(t)}_m, t} = 1$, 
and $\tilde{\delta}_{i^{(t,x)}_m, j^{(t,x)}_m, t}$ is involved in at least one of the violated constraints 
in (\ref{eq_3_7}b) for $x=1$ or (\ref{eq_3_7}c) for $x=2$.
By $\mathfrak{\Delta}_{t,1}$, $t=1,...,T$, and $\mathfrak{\Delta}_{t,2}$, $t=1,...,T$, 
the conflicting links $i \to j@t$ with respect to the half-duplex mode of the UAVs and vessels,
and those with respect to the constraint of the total number of available subcarriers, 
are respectively identified and grouped according to the time slot division. 
It can be inferred that
\begin{align}{\label{eq_3_51}}
	\!\!\! M_{t,x} \leq (I+J)^2.
\end{align}
Note that $\mathfrak{\Delta}_{t,x}, t=1,...,T$, are closely related 
due to the sequential characteristics of the data streaming on the links 
across the time slots $1$ to $T$.
The definition of $\mathfrak{\Delta}_{t,x}$, $t=1,...,T$, $x=1,2$, forms a basis for
the shrinking of the solution space of the relaxed problem~(\ref{eq_3_14})
towards that of the original problem~(\ref{eq_3_7}) following a process-oriented rule.

Secondly, with $\mathfrak{\Delta}_{t,x}$,$t=1,...,T$, $x=1,2$,
we find a smaller solution space for the relaxed problem~(\ref{eq_3_14})
based on the gentlest-ascent principle following a process-oriented rule.
To do this, we form $\bar{M}_x$ sets, i.e., $\bar{\mathfrak{\Delta}}_{x,\bar{m}}, \bar{m}=1,...,\bar{M}_x$, 
based on $ \mathfrak{\Delta}_{t,x}, t=1,...,T$, following a \emph{process-oriented} rule,
respectively for $x=1,2$. 
The specific \emph{process-oriented} rule is delicately designed and the details are illustrated in 
Subsection C.
Specially, for the process-oriented feature, 
$\bar{\mathfrak{\Delta}}_{x,\bar{m}}$ satisfies 
\begin{align}{\label{eq_3_52}}
	| \bar{\mathfrak{\Delta}}_{x,\bar{m}}  \cap \mathfrak{\Delta}_{t,x} | \geq 1,
\end{align}	
for all non-empty $\mathfrak{\Delta}_{t,x}, t=1,...,T$.
For each of the sets $\bar{\mathfrak{\Delta}}_{x,\bar{m}}, \bar{m}=1,...,\bar{M}_x, x=1,2$,
let $\{ \tilde{r}^{(x,\bar{m})}_{i,j,t} \}$ denotes  an optimal solution to the problem in (\ref{eq_3_14})
with a group of extra constraints
\begin{align}{\label{eq_3_53}}
	r_{i,j,t} = 0, ~~(i,j,t) \in \bar{\mathfrak{\Delta}}_{x,\bar{m}},
\end{align}	
which can be written as
\begin{subequations}{\label{eq_3_14_2}}
	\begin{align}
		&\!\!\!\!\!\!\!\!\!\!\!\! \mathop {\min }\limits_{{ \{ r_{i,j,t} \} }}  { \tilde{E}_{\rm{total}} \left( \{r_{i,j,t} \} \right) } \\
		{s.t.} \;\; & S^r_{j,t} \le 1, ~~1 \leq j\leq I+J, 1\leq t \leq T, \\
		\;\;\;\;\;\; &\sum\limits_{i=0}^{I+J'} {\sum\limits_{j=1}^{I+J} { \frac{{r_{i,j,t}}}{ R_{i,j,t} } } } \le N, ~~ 1 \leq t \leq T, \\
		\;\;\;\;\;\; &\left. {{\tilde{V}_{j,t}}} \right|_{t \geq t_j^{{\rm{QoS}}}} \ge V_j^{\rm{QoS}}, ~~I+1 \leq j \leq I+J, \\
		\;\;\;\;\;\; &~{\tilde{V}_{j,t}} \geq \sum\limits_{j'=1}^{I+J} {{r_{j,j',t+1 }}} \Delta \tau, \\
		\;\;\;\;\;\; &~1\leq j \leq I+J',  0\leq t \leq T-1, \nonumber\\
		\;\;\;\;\;\; &~0 \leq r_{i,j,t} \leq R_{i,j,t}, \forall i, j, t, \\
		\;\;\;\;\;\; &~W_{i,j,t} = 1 + \frac{ \beta _{i,j,t } p_{i,j,t} }{ \sigma ^2 + \beta _{i,j,t } p_{i,j,t} W_{i,j,t}^{-1} },\forall i, j, t, \\
		\;\;\;\;\;\; &~r_{i,j,t} = 0, ~~(i,j,t) \in \bar{\mathfrak{\Delta}}_{x,\bar{m}}.
	\end{align}
\end{subequations}
It can be inferred that~(\ref{eq_3_14_2}) is with a smaller solution space compared to~(\ref{eq_3_14}), and
\begin{align}{\label{eq_3_54}}
	\tilde{E}_{total} ( \{ \tilde{r}^{(x,\bar{m})}_{i,j,t} \} ) \geq \tilde{E}_{total} ( \{ \tilde{r}_{i,j,t} \} ).
\end{align}
We set
\begin{align}{\label{eq_3_55}}
\bar{m}_x' = \arg \min_{\bar{m} \in \{ 1,...,\bar{M}_x \}} \tilde{E}_{total} ( \{ \tilde{r}^{(x,\bar{m})}_{i,j,t} \} ).
\end{align}
Then the group of constraints in~(\ref{eq_3_53}) with $\bar{\mathfrak{\Delta}}_{x,\bar{m}} = \bar{\mathfrak{\Delta}}_{x,\bar{m}_x'}$,
together with those in~(\ref{eq_3_14}b)-(\ref{eq_3_14}g),
define a shrinked solution space for the relaxed problem~(\ref{eq_3_14}),
with which the total energy consumption of the network increases the least
compared to the solution spaces formed based on other groups of constraints in~(\ref{eq_3_53}).

Lastly, derive a new group of $\{ \tilde{r}_{i,j,t} \}$ and $\{ \tilde{\delta}_{i,j,t} \}$ based on the relaxed problem~(\ref{eq_3_14})
with the extra group of constraints in~(\ref{eq_3_53}) with $\bar{\mathfrak{\Delta}}_{x,\bar{m}} = \bar{\mathfrak{\Delta}}_{x,\bar{m}_x'}$.
Note that with the group of extra constraints defined by $\bar{\mathfrak{\Delta}}_{x,\bar{m}_x'}$,
the newly obtained $\{ \tilde{r}_{i,j,t} \}$ and $\{ \tilde{\delta}_{i,j,t} \}$,
i.e., the link scheduling and rate adaptation result,
are adjusted across all the time slots $t=1,...,T$ coordinately, i.e., following the process-oriented rule. 
Further,
check the constraints in (\ref{eq_3_7}b) and (\ref{eq_3_7}c).
If any violation is found, then repeat the solution-space-shrinking operation and 
the updating of $\{ \tilde{r}_{i,j,t} \}$ and $\{ \tilde{\delta}_{i,j,t} \}$.
Otherwise, the new group of $\{ \tilde{r}_{i,j,t} \}$ and $\{ \tilde{\delta}_{i,j,t} \}$
constitute a solution for the original problem in~(\ref{eq_3_7}).

\emph{2) Min-Max transformation for $W_{i,j,t}$}

In this part, we try to tackle the second challenge with problem~(\ref{eq_3_7}), 
i.e, the group of hidden non-linear equality constraints in (\ref{eq_3_7}h),
and find an efficient way to solve the relaxed problem in (\ref{eq_3_14}).
Fortunately and interestingly, 
it is found that~(\ref{eq_3_14}) can be transformed into a saddle-point problem
for a \emph{convex-concave} function with a group of linear inequality constraints

Define $f\left(  \{r_{i,j,t} \},  \{z_{i,j,t} \} \right)$ as shown in (32).
\begin{figure*}
	\normalsize
	\setcounter{mytempeqncnt}{\value{equation}}
	\setcounter{equation}{31}
	\begin{equation}{\label{eqn_dbl_x}}
		{f \left( \{r_{i,j,t} \}, \{z_{i,j,t} \} \right) = \sum\limits_{t = 1}^T { {\sum\limits_{i=0}^{I+J'}{\sum\limits_{j = 1}^{I+J} {   \frac{ \sigma^2  }{ \beta_{i,j,t}   } \left( 2^{  \frac{1}{B_s} r_{i,j,t} + \log_2 e ( 1- z_{i,j,t}^{-1} )    } - z_{i,j,t} \right)    }  \Delta \tau  }   }}}.	
	\end{equation}	
	\setcounter{equation}{\value{mytempeqncnt}}
	\hrulefill
\end{figure*}
\setcounter{equation}{32}
The following Theorem~\ref{theorem_min} shows that
the non-linear equality constraints on $W_{i,j,t}$ and $r_{i,j,t}$, $\forall i,j,t$, in 
(\ref{eq_3_7}h) and further (\ref{eq_3_14}g) can be
absorbed through a delicate Min-Max transformation of the problem.

\begin{theorem}\label{theorem_min} 
		For any $r_{i,j,t} \geq 0, \forall i,j,t $, $\tilde{E}_{\rm{total}} \left( \{r_{i,j,t} \} \right)$ can be expressed as
		\begin{eqnarray}\label{eq_3_16}
			\tilde{E}_{\rm{total}} \left( \{r_{i,j,t} \} \right) = \max_{z_{i,j,t} \geq 1, \forall i,j,t }  \,\, f\left(  \{r_{i,j,t} \},  \{z_{i,j,t} \} \right).
		\end{eqnarray}
		Moreover, $f\left(  \{r_{i,j,t} \},  \{z_{i,j,t} \} \right)$ is convex with respect to $\{ r_{i,j,t} \}$ 
		for $r_{i,j,t} \geq 0, \forall i,j,t$, 
		and concave with respect to $\{ z_{i,j,t} \}$ for $z_{i,j,t} \geq 1, \forall i,j,t$.

\end{theorem}

\begin{proof}
	See Appendix A.
\end{proof}

Based on Theorem~\ref{theorem_min}, the problem in (\ref{eq_3_14}) can be equivalently recast as
the following Min-Max problem 
\begin{subequations}{\label{eq_3_17}}
	\begin{align}
		&\!\!\!\!\!\!\!\!\!\!\!\! \mathop {\min }\limits_{{ \{ r_{i,j,t} \} }} \max_{ \{ z_{i,j,t} \} }  { f\left(  \{r_{i,j,t} \},  \{z_{i,j,t} \} \right) } \\
		{s.t.} \;\; & S^r_{j,t} \le 1,  ~~1 \leq j\leq I+J, 1\leq t \leq T, \\
		\;\;\;\;\;\; &\sum\limits_{i=0}^{I+J'} {\sum\limits_{j=1}^{I+J} { \frac{{r_{i,j,t}}}{ R_{i,j,t} } } } \le N, ~~1\leq t \leq T, \\
		\;\;\;\;\;\; &\left. {{\tilde{V}_{j,t}}} \right|_{t \geq t_j^{{\rm{QoS}}}} \ge V_j^{\rm{QoS}}, ~~I+1 \leq j \leq I+J, \\
		\;\;\;\;\;\; &~{\tilde{V}_{j,t}} \geq \sum\limits_{j'=1}^{I+J} {{r_{j,j',t+1 }}} \Delta \tau, \\
		\;\;\;\;\;\; &~1\leq j \leq I+J', 0\leq t \leq T-1, \nonumber\\
		\;\;\;\;\;\; &~0 \leq r_{i,j,t} \leq R_{i,j,t}, \forall i, j, t, \\
		\;\;\;\;\;\; &~ z_{i,j,t} \geq 1, \forall i, j, t.
	\end{align}
\end{subequations}
With a \emph{convex-concave} objective function
and a group of linear inequality constraints, 
(\ref{eq_3_17}) can be solved via computations with a polynomial complexity~\cite{SaddlePoint_1, SaddlePoint_2, SaddlePoint_3}. 	
Its optimal solution is a saddle point
of $f\left(  \{r_{i,j,t} \},  \{z_{i,j,t} \} \right)$ within the convex set
determined by the constraints (\ref{eq_3_17}b)-(\ref{eq_3_17}g).
Accordingly, an optimal solution can be obtained for the problem in (\ref{eq_3_14}). 

\begin{figure}[t]
	\centering
	\includegraphics[width=8cm]{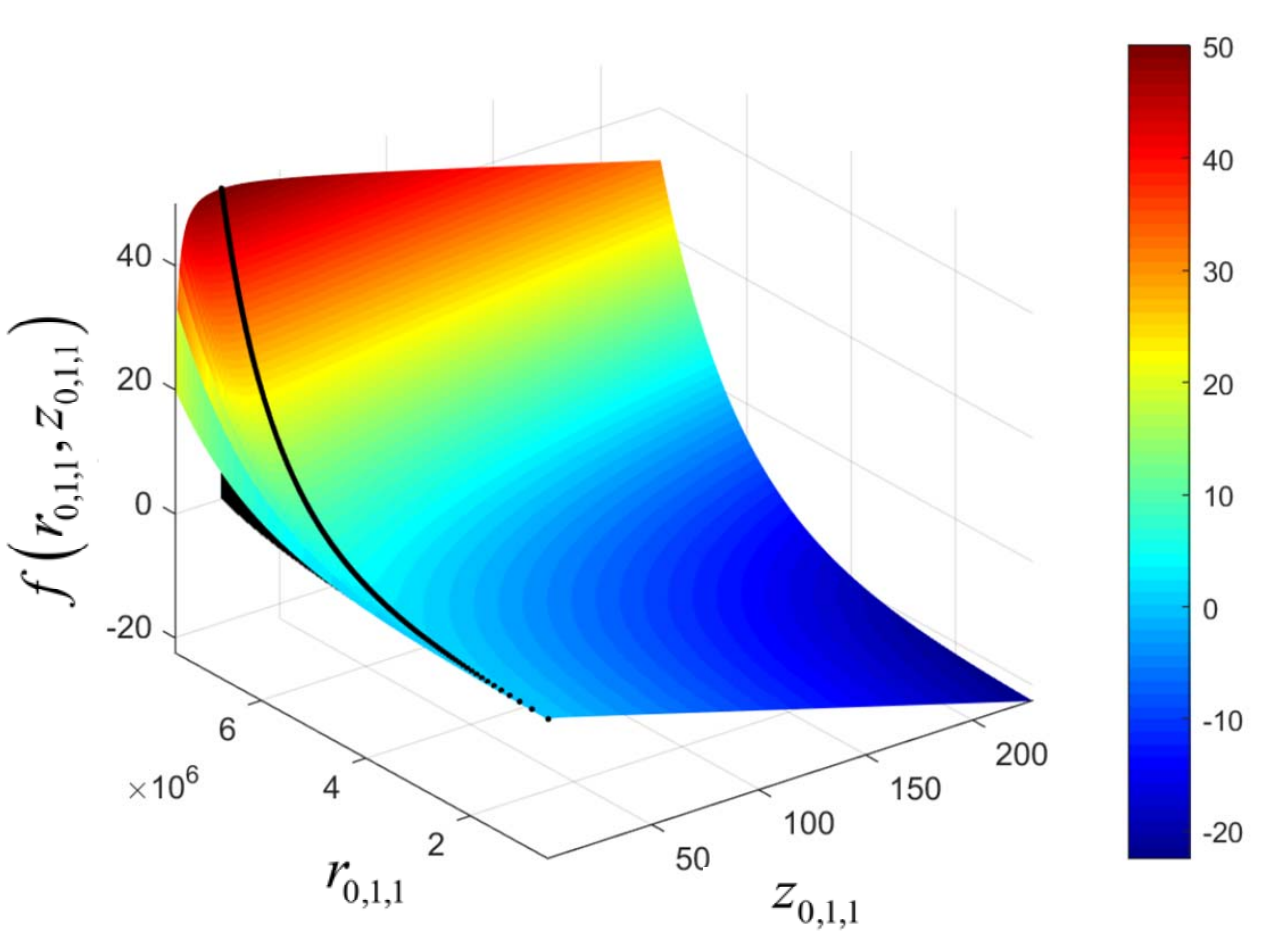}
	\caption{Variation of $f\left(  r_{0,1,1} ,  z_{0,1,1} \right)$ with respect to $[r_{0,1,1} ,  z_{0,1,1}]$ .}
	\label{fig_function_1}
\end{figure}

\begin{figure}[t]
	\centering
	\includegraphics[width=8cm]{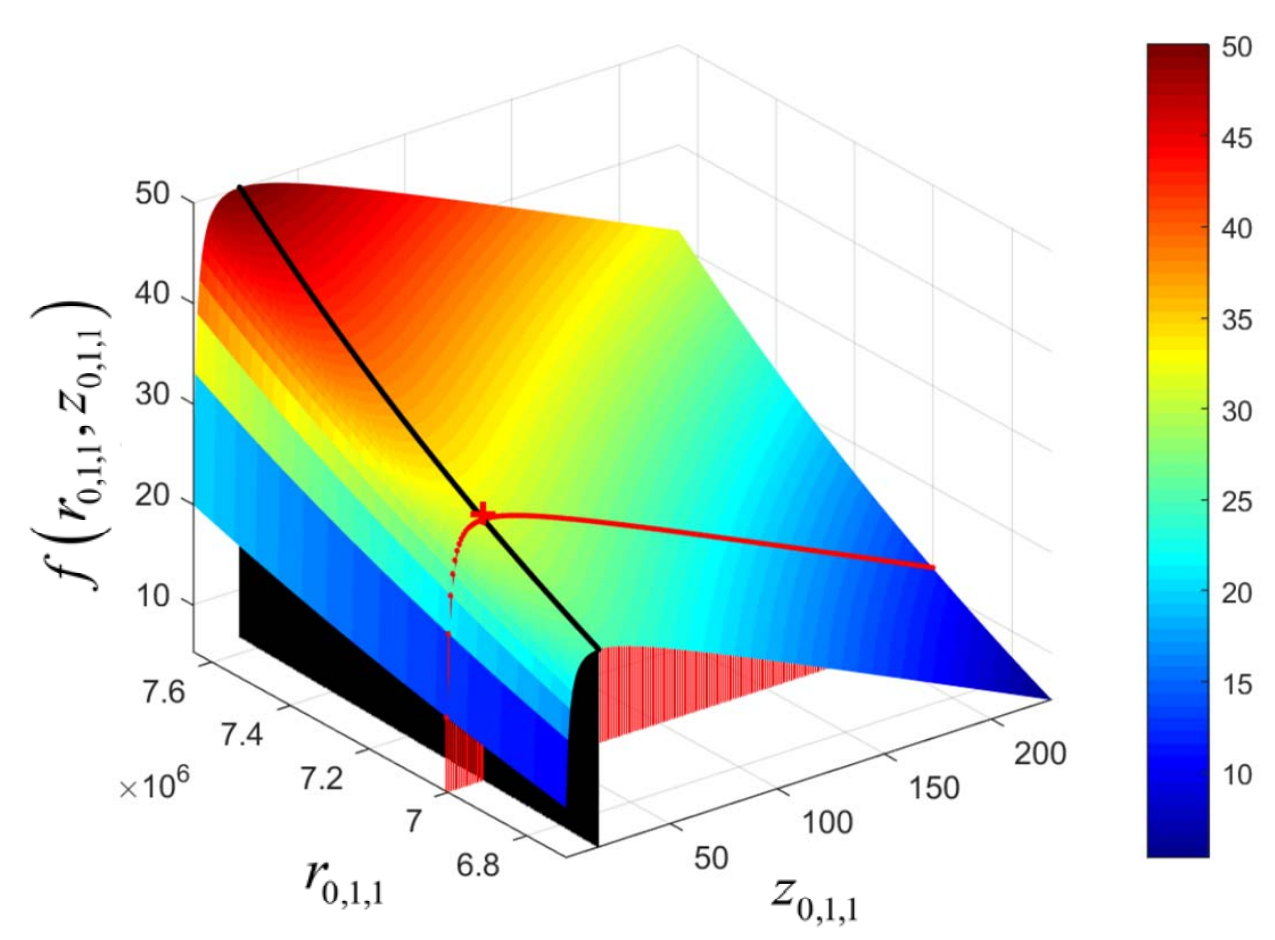}
	\caption{Illustration of $f\left(  r_{0,1,1} ,  z_{0,1,1} \right)$ with more details.}
	\label{fig_function_2}
\end{figure}

To show the characteristics of $f\left(  \{r_{i,j,t} \},  \{z_{i,j,t} \} \right)$ more clearly,
we consider a simplified network scenario 
with $I=0,J=1,T=1,N=1$, $\Delta t = 1$s, $B_s = 1$MHz, $h_0 = 30$m, $h_1 = 5$m, $d_{0,1,1}=100$m, $f_c = 2$GHz, $C=1$,
$P_0 = 50$W, and $\sigma^2 = -114$dBm.
With $p_{0,1,1} = P_0$, we get $R_{0,1,1} = 7.65$Mbps.
In this scenario, the variation of $f\left(  r_{0,1,1} ,  z_{0,1,1} \right)$ with respect to 
$[r_{0,1,1} ,  z_{0,1,1}]$ within 
$[0, R_{0,1,1}] \times [1,+\infty)]$ is shown in Fig.~\ref{fig_function_1}.
For more details, $f\left(  r_{0,1,1} ,  z_{0,1,1} \right)$ 
within 
$[7/8R_{0,1,1}, R_{0,1,1}] \times [1,+\infty)$
is further illustrated in Fig.~\ref{fig_function_2}.
The black line on the surface in both Fig.~\ref{fig_function_1} and Fig.~\ref{fig_function_2} marks all the points with $z_{0,1,1}=W_{0,1,1}$,
and the red \textbf{+} marker on the surface in Fig.~\ref{fig_function_2} indicates  
the saddle point of $f\left(  r_{0,1,1} ,  z_{0,1,1} \right)$
with constraints $ r_{0,1,1} \geq 10/11 R_{0,1,1} $.

\subsection{Suboptimal Joint Link Scheduling and Rate Adaptation Scheme}
Based on the analysis in Subsection B, we propose a suboptimal joint link scheduling and rate adaptation scheme for the 
hybrid satellite-UAV-terrestrial MCN.

The details of the proposed scheme is illustrated in Algorithm~\ref{Iterative_scheme}.
Note that the solution space for the relaxed problem in~(\ref{eq_3_14})
is shrinked with the help of $\mathfrak{\Delta}_{t,1}, t=1,...,T$, first,
which eliminate the violations of (\ref{eq_3_7}b), 
and then $\mathfrak{\Delta}_{t,2}, t=1,...,T$, which eliminate that of (\ref{eq_3_7}c).

For $x=1,2$,
in each iteration $s_x$, $\bar{M}_x = M_{\bar{T}_x,x}$ sets, i.e., $ \bar{\mathfrak{\Delta}}_{x,\bar{m}}, \bar{m}=1,...,M_{\bar{T}_x,x}$,
are formed based on $\left[ \mathfrak{\Delta}_{t,x} \right]^{s_x}, t=1,...,T$.
$\bar{T}_x$ is the biggest $t \in \{ 1,...,T \}$ with $\left[ \mathfrak{\Delta}_{t,x} \right]^{s_x} \neq \emptyset $,
and $M_{\bar{T}_x,x}$ is the size of the set $\left[ \mathfrak{\Delta}_{\bar{T}_x,x} \right]^{s_x}$.
More specifically, $ \bar{\mathfrak{\Delta}}_{x,\bar{m}}, \bar{m}=1,...,M_{\bar{T}_x,x}$, are derived 
following the \emph{process-oriented} rule
defined by (\ref{eq_30})-(\ref{eq_34}), as
\begin{align}{\label{eq_30}}
	\bar{\mathfrak{\Delta}}_{x,\bar{m}} = \cup_{t=1}^{\bar{T}_x} \bar{\mathfrak{\Delta}}'_{x,\bar{m},t},
\end{align}
where 	
\begin{subequations} {\label{eq_30_1}}
	\begin{align}
		&\bar{\mathfrak{\Delta}}'_{x,\bar{m},t} = \cup_{k=1}^{4} \hat{\mathfrak{\Delta}}^{(k)}_{x,\bar{m},t},~~~~~~~~~~~~~~ x=1, \\
		&\bar{\mathfrak{\Delta}}'_{x,\bar{m},t} = \{  (i^{(t,x)}_{\hat{m}_{x,\bar{m},t}}, j^{(t,x)}_{\hat{m}_{x,\bar{m},t}},t)  \}, ~~ x=2,
	\end{align}
\end{subequations}
with\footnote{ Note that some elements in the $\hat{\mathfrak{\Delta}}^{(k)}_{x,\bar{m},t}$, $k=1,...,4$, 
	given in~(\ref{eq_31}a)-(\ref{eq_31}d) may be invalid in our considered network, as the BS only transmits and
	the vessels $I+J'+1,..., I+J$ only receives. When an invalid element appears, we could just ignore it.
	It does not affect the effectiveness of the expression in~(\ref{eq_31}a)-(\ref{eq_31}d).} 
\begin{subequations}{\label{eq_31}}
	\begin{align}
		&\!\!\!\!  \hat{\mathfrak{\Delta}}^{(1)}_{x,\bar{m},t} = \nonumber \\
		&\!\!\!\!\{ (i,j,t) | i=i^{(t,x)}_{\hat{m}_{x,\bar{m},t}}; j\in\{1,...,I+J\}, j \neq j^{(t,x)}_{\hat{m}_{x,\bar{m},t}} \}, \\
		&\!\!\!\! \hat{\mathfrak{\Delta}}^{(2)}_{x,\bar{m},t} = \{ (i,j,t) | i \in \{0,...,I+J'\}; j=i^{(t,x)}_{\hat{m}_{x,\bar{m},t}} \}, \\
		&\!\!\!\! \hat{\mathfrak{\Delta}}^{(3)}_{x,\bar{m},t}= \nonumber \\
		&\!\!\!\! \{ (i,j,t) | i\in\{0,...,I+J'\}, i \neq i^{(t,x)}_{\hat{m}_{x,\bar{m},t}}; j=j^{(t,x)}_{\hat{m}_{x,\bar{m},t}} \}, \\
		&\!\!\!\! \hat{\mathfrak{\Delta}}^{(4)}_{x,\bar{m},t} = \{ (i,j,t) |  i = j^{(t,x)}_{\hat{m}_{x,\bar{m},t}}; j \in \{1,...,I+J\} \}.
	\end{align} 
\end{subequations}
In~(\ref{eq_30_1}b) and (\ref{eq_31}a)-(\ref{eq_31}d),
$(i^{(t,x)}_{\hat{m}_{x,\bar{m},t}}, j^{(t,x)}_{\hat{m}_{x,\bar{m},t}},t)$ is the $\hat{m}_{x,\bar{m},t}$th element of
the non-empty set $\left[ \mathfrak{\Delta}_{t,x} \right]^{s_x}$.\footnote{If $\left[ \mathfrak{\Delta}_{t,x} \right]^{s_x} = \emptyset$,
	we could just ignore the corresponding $\bar{\mathfrak{\Delta}}'_{x,\bar{m},t}$ in (\ref{eq_30})
	as well as all subsequent relative expressions.}
When $t= \bar{T}_x$, 
\begin{align}{\label{eq_30_2}}
	\hat{m}_{x,\bar{m},t} = \bar{m}, ~~\bar{m} = 1,...,M_{\bar{T}_x,x},
\end{align}
and for $1 \leq t \leq \bar{T}_x-1$, $\hat{m}_{x,\bar{m},t}$ is determined by
\begin{align}{\label{eq_32}}
	\hat{m}_{x,\bar{m},t} = \arg \min_{m \in \{ 1,...,M_{t,x} \}}  \tilde{E}_{total} ( \{ \hat{r}^{(x,m,\bar{m})}_{i,j,t} \} ),
\end{align}
where $M_{t,x}$ is the size of the set $\left[ \mathfrak{\Delta}_{t,x} \right]^{s_x}$, and
$\{ \hat{r}^{(x,m,\bar{m})}_{i,j,t} \} $ is the solution to (\ref{eq_3_17}) with the following extra constraints as
\begin{align}{\label{eq_33}}
	r_{i,j,t} = 0, ~(i,j,t) \in \bar{\mathfrak{\Delta}}''_{x,m,t} \cup_{t=t+1}^{\bar{T}_x} \bar{\mathfrak{\Delta}}'_{x,\bar{m},t} \cup  \left[ \bar{\mathfrak{\Delta}} \right]^{s_x},
\end{align}
with 
\begin{subequations}{\label{eq_34}}
	\begin{align}
		& \bar{\mathfrak{\Delta}}''_{x,m,t} = \cup_{k=1}^{4} \hat{\mathfrak{\Delta}}'^{(k)}_{x,m,t},~~x=1, \\
		& \bar{\mathfrak{\Delta}}''_{x,m,t} = \{  (i^{(t,x)}_m, j^{(t,x)}_m, t) \}, ~~x=2,
	\end{align}
\end{subequations}
in which $ (i^{(t,x)}_m, j^{(t,x)}_m, t)$ is the $m$th element in $\left[ \mathfrak{\Delta}_{t,x} \right]^{s_x}$,
and $\hat{\mathfrak{\Delta}}'^{(k)}_{x,m,t}, k=1,...,4$, can be respectively obtained from~(\ref{eq_31}a)-(\ref{eq_31}d)  by 
replacing $(i^{(t,x)}_{\hat{m}_{x,\bar{m},t}}, j^{(t,x)}_{\hat{m}_{x,\bar{m},t}},t)$
with $ (i^{(t,x)}_m, j^{(t,x)}_m, t)$.

\begin{algorithm}[t]
	\caption{Suboptimal joint link scheduling and rate adaptation scheme.}
	\begin{algorithmic}[1]
		\label{Iterative_scheme}
		\STATE Solve the problem in (\ref{eq_3_17}), and denote the solution as  ${\tiny \{ \tilde{r}_{i,j,t} \} }$. Detach ${\tiny \{ \tilde{\delta}_{i,j,t} \} }$ from ${\tiny \{ \tilde{r}_{i,j,t} \} }$ based on (\ref{eq_3_11}). 
		\STATE With $\{ \tilde{\delta}_{i,j,t} \}$, initialize 
		$\left[ \mathfrak{\Delta}_{t,1} \right]^0, t=1,...,T $, 
		based on (\ref{eq_3_50}).
		\STATE Set
		$\left[ \bar{\mathfrak{\Delta}} \right]^0 = \emptyset $, and $s_1 = 0$.
		
		\FOR{x=1,2}
		
		\WHILE{$ \cup_{t=1}^{T} \left[ \mathfrak{\Delta}_{t,x} \right]^{s_x} \neq \emptyset $ }
		\STATE Determine $\bar{T}_x$ based on $\left[ \mathfrak{\Delta}_{t,x} \right]^{s_x}, t=1,...,T$, by finding the biggest $t \in \{1,...,T\}$ with $\left[ \mathfrak{\Delta}_{t,x} \right]^{s_x} \neq \emptyset $.
		
		\STATE Derive $ \bar{\mathfrak{\Delta}}_{x,\bar{m}}, \bar{m}=1,...,M_{\bar{T}_x,x}$, based on 
		$\left[ \mathfrak{\Delta}_{t,x} \right]^{s_x}, t=1,...,T$, 
		according to  (\ref{eq_30})-(\ref{eq_34}), through iteratively solving (\ref{eq_3_17}) with the extra constraints 
		given in (\ref{eq_33}).
		\STATE Set $ \bar{\mathfrak{\Delta}}_{x,\bar{m}} ~:=~ \bar{\mathfrak{\Delta}}_{x,\bar{m}} \cup  \left[ \bar{\mathfrak{\Delta}} \right]^{s_x }, \bar{m}=1,...,M_{\bar{T}_x,x}$.
		\STATE Calculate $\{ \tilde{r}^{(x,\bar{m})}_{i,j,t} \}$, $\bar{m}=1,...,M_{\bar{T}_x,x}$, by solving (\ref{eq_3_17}) with the extra constraints in (\ref{eq_3_53}).
		\STATE Find $\bar{m}_x' = \arg \min_{\bar{m} \in \{ 1,...,M_{\bar{T}_x,x} \}} \tilde{E}_{total} ( \{ \tilde{r}^{(x,\bar{m})}_{i,j,t} \} )$.
		\STATE Set $\left[ \bar{\mathfrak{\Delta}} \right]^{s_x+1} =  \bar{\mathfrak{\Delta}}_{ x, \bar{m}_x'}$.
		\STATE Update $\{ \tilde{r}_{i,j,t} \}$ as the solution of (\ref{eq_3_17}) with a group of extra constraints as
		$r_{i,j,t}=0, (i,j,t) \in \left[ \bar{\mathfrak{\Delta}} \right]^{s_x+1}$, 
		and update the corresponding $\{ \tilde{\delta}_{i,j,t} \} $ based on (\ref{eq_3_11}).
		\STATE Derive  $\left[ \mathfrak{\Delta}_{t,x} \right]^{s_x+1}, t=1,...,T $, 
		based on $\{ \tilde{\delta}_{i,j,t} \} $ according to (\ref{eq_3_50}).
		\STATE $s_x := s_x + 1$. 
		\ENDWHILE 
		
		\IF{x=1}
		\STATE Update $\left[ \bar{\mathfrak{\Delta}} \right]^0 = \left[ \bar{\mathfrak{\Delta}} \right]^{s_1} $, and set $x=2, s_2 = 0$.
		\STATE Initialize  $\left[ \mathfrak{\Delta}_{t,2} \right]^{0}, t=1,...,T $, 
		based on $\{ \tilde{\delta}_{i,j,t} \} $ according to (\ref{eq_3_50}).
		\ENDIF
		
		\ENDFOR
		
		\STATE Set $\bar{\mathfrak{\Delta}} = \left[ \bar{\mathfrak{\Delta}} \right]^{s_2}$. 
	\end{algorithmic}
\end{algorithm}

Note that 
by constraining all $r_{i,j,t}$ with
$(i,j,t) \in \bar{\mathfrak{\Delta}}'_{x,\bar{m},t}$ determined by (\ref{eq_30_1}a) and (\ref{eq_31})
to $0$ at the same time,
it is assured that the constraints in (\ref{eq_3_7}b) are satisfied for both UAV or vessel $i^{(t,x)}_{\hat{m}_{x,\bar{m},t}}$ and $j^{(t,x)}_{\hat{m}_{x,\bar{m},t}}$ in time slot $t$ when $i^{(t,x)}_{\hat{m}_{x,\bar{m},t}} \geq 1$,
or for UAV or vessel $j^{(t,x)}_{\hat{m}_{x,\bar{m},t}}$ when $i^{(t,x)}_{\hat{m}_{x,\bar{m},t}} = 0$.
Further, the derivation of $ \bar{\mathfrak{\Delta}}_{x,\bar{m}}, \bar{m}=1,...,M_{\bar{T}_x,x}$, $x=1,2$,
in (\ref{eq_30}), can be efficiently carried out by determining $\bar{\mathfrak{\Delta}}'_{x,\bar{m},t}$
based on a reverse sequence with respect to $t$, i.e., from $\bar{\mathfrak{\Delta}}'_{x,\bar{m},\bar{T}_x}$
to $\bar{\mathfrak{\Delta}}'_{x,\bar{m},1}$,
through repeatedly solving the problem in~(\ref{eq_3_17}) with the extra constraints given in~(\ref{eq_33}).
Besides, $\left[ \bar{\mathfrak{\Delta}} \right]^{s_x}$ in Algorithm~\ref{Iterative_scheme}
is used to accumulate the extra \emph{forced-to-zero} constraints on $r_{i,j,t}$, as shown in~(\ref{eq_3_53}),
along with the iterations,
so as to ensure that the solution space for the relaxed problem~(\ref{eq_3_14}) can be continuingly shrinked
towards that of the original problem~(\ref{eq_3_7}).

Besides, to solve 
the problem (\ref{eq_3_17}) with a group of extra constraints determined by any set $\bar{\mathfrak{\Delta}}$ as
\begin{align}{\label{eq_3_30}}
	r_{i,j,t} = 0, ~~(i,j,t) \in \bar{\mathfrak{\Delta}},
\end{align}
we can delete all the $r_{i,j,t}$ in (\ref{eq_3_30})
as well as the corresponding $z_{i,j,t}$
from its objective function and constraints,
and solve the following simplified problem as
\begin{subequations}{\label{eq_3_31}}
	\begin{align}
		&\!\!\!\!\!\!\!\!\!\!\!\! \mathop {\min }\limits_{{ \{ r_{i,j,t} \} }} \max_{ \{ z_{i,j,t} \} }  { f_{\bar{\mathfrak{\Delta}}} \left(  \{r_{i,j,t} \},  \{z_{i,j,t} \} \right) } \\
		{s.t.} \;\; &S^{r(\bar{\mathfrak{\Delta}})}_{j,t} \le 1,  ~~1\leq j \leq I+J, 1\leq t \leq T, \\
		\;\;\;\;\;\; &\underset{{(i,j,t) \not\in \bar{\mathfrak{\Delta}}}} {\sum\limits_{i=0}^{I+J'} \sum\limits_{j=1}^{I+J}}   { \frac{{r_{i,j,t}}}{ R_{i,j,t} } }  \le N, ~~1\leq t \leq T, \\
		\;\;\;\;\;\; &\left. {{\tilde{V}_{j,t}}}^{(\bar{\mathfrak{\Delta}})} \right|_{t \geq t_j^{{\rm{QoS}}}} \ge V_j^{\rm{QoS}}, ~~I+1 \leq j \leq I+J, \\
		\;\;\;\;\;\; &~{\tilde{V}_{j,t}}^{(\bar{\mathfrak{\Delta}})} \geq \sum\limits_{j'=1, {(j,j',t+1) \not\in \bar{\mathfrak{\Delta}}}}^{I+J} {{r_{j,j',t+1 }}} \Delta \tau, \\
		\;\;\;\;\;\; &~ 1\leq j \leq I+J',   0\leq t \leq T-1, \nonumber\\
		\;\;\;\;\;\; &~0 \leq r_{i,j,t} \leq R_{i,j,t}, {(i,j,t) \not\in \bar{\mathfrak{\Delta}}}, \\
		\;\;\;\;\;\; &~z_{i,j,t} \geq 1, {(i,j,t) \not\in \bar{\mathfrak{\Delta}}}.
	\end{align}
\end{subequations}
where $f_{\bar{\mathfrak{\Delta}}} \left(  \{r_{i,j,t} \},  \{z_{i,j,t} \} \right)$, $S^{r(\bar{\mathfrak{\Delta}})}_{j,t}$, and 
${\tilde{V}_{j,t}}^{(\bar{\mathfrak{\Delta}})}$
are respectively given by (44), (\ref{eq_3_31_1}), and (\ref{eq_3_31_2}).
\begin{figure*}
	\normalsize
	\setcounter{mytempeqncnt}{\value{equation}}
	\setcounter{equation}{43}
	\begin{equation}{\label{eqn_dbl_x}}
		{f \left( \{r_{i,j,t} \}, \{z_{i,j,t} \} \right) = \underset{{(i,j,t) \not\in \bar{\mathfrak{\Delta}}}} { \sum\limits_{t = 1}^T     \sum\limits_{i=0}^{I+J'}\sum\limits_{j = 1}^{I+J}  } {   \frac{ \sigma^2  }{ \beta_{i,j,t}   } \left( 2^{  \frac{1}{B_s} r_{i,j,t} + \log_2 e ( 1- z_{i,j,t}^{-1} )    } - z_{i,j,t} \right)    }  \Delta \tau      }.	
	\end{equation}	
	\setcounter{equation}{\value{mytempeqncnt}}
	\hrulefill
\end{figure*}
\setcounter{equation}{44}
\begin{equation}{\label{eq_3_31_1}}
	S^{r(\bar{\mathfrak{\Delta}})}_{j,t} = \left\{
	\begin{aligned}
		& \underset{{(i,j,t) \not\in \bar{\mathfrak{\Delta}}}} {\sum\limits_{i=0}^{I+J'}} { \frac{r _{i,j,t}}{ R_{i,j,t} } } + \underset{{(j,j',t) \not\in \bar{\mathfrak{\Delta}}}} {\sum\limits_{j'=1}^{I+J}} { \frac{r _{j,j',t}}{ R_{j,j',t} } }, \\
		&~~~~~~~~~~~~~~~~~~~~~~~~~~~~~j\in \{1,...,I+J'\}, \\
		& \underset{{(i,j,t) \not\in \bar{\mathfrak{\Delta}}}} {\sum\limits_{i=0}^{I+J'}} { \frac{r _{i,j,t}}{ R_{i,j,t} } }, ~~j \in \{I+J'+1,...,I+J\}.
	\end{aligned}
	\right.
\end{equation}
\begin{equation}{\label{eq_3_31_2}}
	\tilde{V}_{j,t}^{(\bar{\mathfrak{\Delta}})} = \left\{
	\begin{aligned}
		& \underset{{(i,j,t) \not\in \bar{\mathfrak{\Delta}}}} { \sum\limits_{i=0}^{I+J'}} {{r_{i,j, t }}} \Delta \tau,  ~~~t=1, \forall j, \\
		& \sum\limits_{\tau = 1}^t {\left( { \underset{{(i,j,\tau) \not\in \bar{\mathfrak{\Delta}}}} { \sum\limits_{i=0}^{I+J'}} {{r_{i,j,\tau }}} - \underset{{(j,j',\tau) \not\in \bar{\mathfrak{\Delta}}}} {\sum\limits_{j'=1}^{I+J}} {{r_{j,j',\tau }}} } \right)
			\Delta \tau }, \\
		&~~~2 \leq t \leq T, 1\leq j \leq I+J', \\
		& \underset{{(i,j,\tau) \not\in \bar{\mathfrak{\Delta}}}} {\sum\limits_{\tau = 1}^t  \sum\limits_{i=0}^{I+J'} } {{r_{i,j,\tau }}} 
		\Delta \tau ,\\
		&~~~2 \leq t \leq T, I+J'+1 \leq j \leq I+J.
	\end{aligned}
	\right.
\end{equation}

With the finally obtained $\bar{\mathfrak{\Delta}} $ by Algorithm~\ref{Iterative_scheme},
solve the Min-Max problem in (\ref{eq_3_17}) with a group of extra constraints as
\begin{equation}{\label{eq_3_20}}
	r_{i,j,t} = 0, ~~(i,j,t) \in \bar{\mathfrak{\Delta}},
\end{equation}
and denote the optimal solution as $\{ r^{op}_{i,j,t} \}$.
Further, set $\{ \delta^{op}_{i,j,t} \}$  as 
\begin{equation}{\label{eq_3_21}}
	\tilde{\delta}_{i,j,t} = \left\{
	\begin{aligned}
		& 1,  (i,j,t)  \not\in  \bar{\mathfrak{\Delta}} , \\
		& 0,  (i,j,t) \in  \bar{\mathfrak{\Delta}}.
	\end{aligned}
	\right.
\end{equation}
Then, $\{ r^{op}_{i,j,t} \}$ and $\{ \delta^{op}_{i,j,t} \}$ constitute 
a solution for the original joint link scheduling and rate adaptation problem shown in~(\ref{eq_3_7})
for the hybrid satellite-UAV-terrestrial MCN.

The following Theorem~\ref{theorem_converge} shows that 
the proposed joint link scheduling and rate adaptation scheme in Algorithm~\ref{Iterative_scheme}
always converges within $s \leq [2(I+J)-N]T$ iterations.
In later simulations, we will show that the proposed scheme achieves a comparable performance to the optimal solution
obtained via exhaustive search, while it can be effectively implemented in practice with much lower complexity.

\begin{theorem}\label{theorem_converge} 
The proposed joint link scheduling and rate adaptation scheme in Algorithm~\ref{Iterative_scheme}	
converges with a maximum iteration number of $[2(I+J)-N]T$,
with $s_1 \leq (I+J)T$ and $s_2 \leq [(I+J)-N]T$.
\end{theorem}

\begin{proof}
	See Appendix B.
\end{proof}

\subsection{Complexity Analysis}
The complexity of the proposed scheme
is mainly caused by the iterative solving of the Min-Max problems in (\ref{eq_3_17}).
Denote the number of iterations needed for the implementation of Algorithm~\ref{Iterative_scheme}
for $x=1$ and $x=2$ as $\mathfrak{I}_1$ and $\mathfrak{I}_2$, respectively.
Based on Theorem~\ref{theorem_converge}, an upper bound for $\mathfrak{I}_1$, $\mathfrak{I}_2$ can be respectively written as 
\begin{subequations}{\label{eq_3_60}}
	\begin{align}
		&\mathfrak{I}_1 \leq \bar{\mathfrak{I}}_1  =  (I+J)T, \\
		&\mathfrak{I}_2 \leq \bar{\mathfrak{I}}_2  =  [(I+J)-N]T.
	\end{align}
\end{subequations}
Further, an upper bound for the number of problems solved in each iteration $s_x$, i.e., $N_{s,x}, x=1,2$, 
can be written as
\begin{subequations}{\label{eq_3_34}}
	\begin{align}
		&\!\!\!\!\! N_{s,1} \leq \bar{N}_{s,1} = (I+J)^2(T-1) (I+J) \nonumber \\
		&~~~~~~~~~~~~ =(I+J)^3(T-1), \\
		&\!\!\!\!\! N_{s,2} \leq \bar{N}_{s,2} = (I+J)(T-1)(I+J) \nonumber \\
		&~~~~~~~~~~~~ = (I+J)^2(T-1).
	\end{align}
\end{subequations}
Overall, the total number of problems that need to be solved, i.e., $\mathfrak{N}$, can be written as 
\begin{align}{\label{eq_3_61}}
	\mathfrak{N} &= \sum_{s_1=0}^{\mathfrak{I}_1-1} N_{s,1} + \sum_{s_2=0}^{\mathfrak{I}_2-1} N_{s,2} \leq 
	\sum_{s_1=0}^{\bar{\mathfrak{I}}_1-1} \bar{N}_{s,1} + \sum_{s_2=0}^{\bar{\mathfrak{I}}_2-1} \bar{N}_{s,2} \nonumber \\
	&  = (I+J)^2(T-1)T [ (I+J)^2 + I+J -N ] \nonumber \\
	& \overset{\Delta}{=}\bar{\mathfrak{N}}.
\end{align}

As (\ref{eq_3_17}) can be solved with polynomial complexity, 
the proposed scheme is also with a polynomial complexity according to~(\ref{eq_3_61}).
Actually, $\bar{\mathfrak{N}} $ is a quite loose upper bound for $\mathfrak{N}$. 
Simulation results show that for the network scenarios considered in Section IV, 
$\mathfrak{N}$ is less than $1/100$ of $\bar{\mathfrak{N}}$.

Thus, though suboptimal, 
the proposed joint link scheduling and rate adaptation scheme has a much lower complexity
compared to the optimal scheme realized via exhaustive search, 
for which a maximum computation complexity of $O\left( 2^{\left( I+J \right)^2T} \right)$ is required,
which rapidly becomes prohibitive in practice for relatively large $I+J$ and $T$.

\section{Simulation Results and Discussions}
In the simulations, we consider a MCN with $I=1$, $J=9$, $J'=8$, $T=10$, and $N=9$, deployed within a square area of 25km$^2$.
The on-shore BS is located on the edge of the area, and the vessels are assumed to sail along 
randomly-generated straight shipping lanes. 
A randomly sampled topology of the network is shown in Fig.~\ref{fig2r}.

\begin{figure}[t]
	\centering
	\includegraphics[width=8cm]{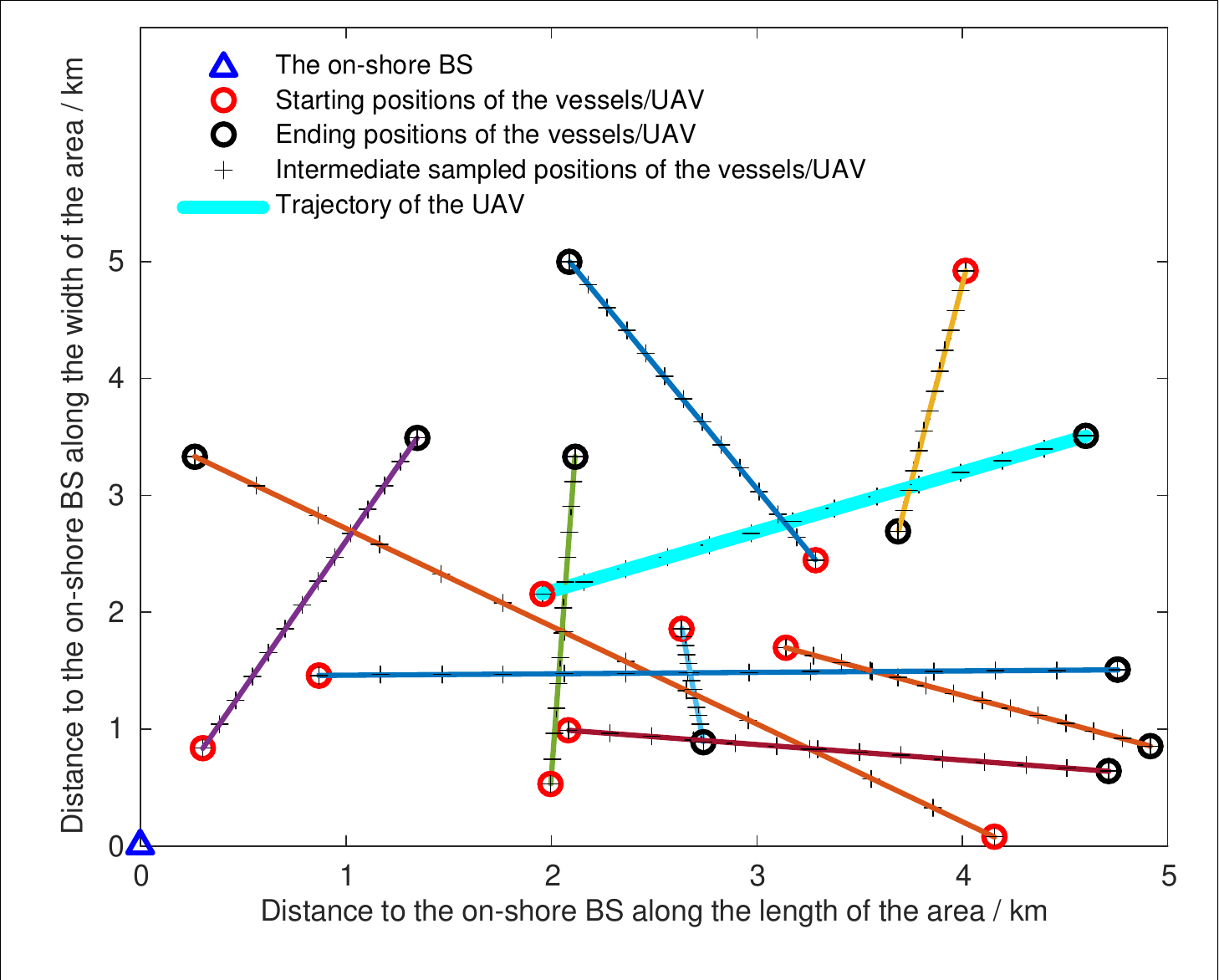}
	\caption{A randomly sampled topology of the hybrid MCN.}
	\label{fig2r}
\end{figure}

Each time slot lasts $\Delta \tau = 30$s, and the bandwidth of each subcarrier is $B_s = 1$MHz.
The height of the BS is $50$m, the height of the UAV is $100$m,
and the antenna heights of the vessels are all set as $5$m, i.e., $h_i = 5$m, $i=2,...,10$.
The carrier frequency is set as $f_c = 2$GHz, and $C=1$.
Considering that the UAV flies at a relatively low height, 
we set the channel parameters for the links from or to the UAV as
$a=5.0188$, $b=0.3511$, $\eta_{LOS}=2.3$, $\eta_{NLOS}=34$, respectively.
The maximum transmit power of the BS is $P_0 = 50$W, the maximum transmit power of the UAV is $P_1 = 10$W,
and that of each of the first $J'=8$ vessels is $P_i = 10$W, $i=2,...,9$.
The noise power is $\sigma^2 = -114$dBm.
Without loss of generality, the QoS guarantees are considered to be in proportion to the transmission rates of the BS to the vessels,
i.e, $V_j^{QoS} = \alpha \sum_{\tau=1}^T R_{0,j,\tau} \Delta \tau, j=2,...,10$, with $ 0 < \alpha \leq 1$.
$t_j^{QoS}$ is set as $t_j^{QoS} = T$ for $j=2,...,8$, and $t_j^{QoS} = T-1$ or $j=9,10$.

\subsection{Performance of the Proposed Scheme}

The energy consumption of the network 
with the proposed joint link scheduling and rate adaptation scheme for
$4$ different groups of QoS guarantees, i.e., $\alpha = 1/4, 1/3, 1/2$, $2/3$, is shown in Fig.~\ref{fig3}.
Note that the network energy consumption in Fig.~\ref{fig3} is obtained based on the average among that for $10$ randomly sampled 
topologies of the hybrid MCN. 
For comparison, we also show the performance of two other schemes.
The first one is the fixed transmission scheme, in which the BS transmits directly to the vessels
with the maximum transmit power $P_0$ in a certain number of time slots with relatively large $R_{i,j,t}$ 
so as to satisfy the QoS guarantees.
More specifically, the BS selects $T'_j$ time slots to transmit to vessel $j$, where $T'_j$ satisfies
$\sum_{y=1}^{T'_j} R_{0,j,\tau_y} \Delta \tau \geq V_j^{QoS}  $ and $\sum_{y=1}^{T'_j-1} R_{0,j,\tau_y} \Delta \tau < V_j^{QoS}  $,
with $R_{0,j,\tau_y}, y=1,...,T'_j$, being the largest $T'_j$ ones among $R_{0,j,t}, t=1,...,T$.
The second scheme is the rate-adaptation-only scheme developed when only the links between the BS and the vessels,
i.e., $0 \to j@t$, $j \in \{I+1,...,I+J\}, t \in \{1,...,T\}$, can be utilized,
by sovling the problem  (\ref{eq_3_14}) or (\ref{eq_3_17})  with the extra constraints as
$r_{i,j,t} = 0, i \in \{ 1,...,I+J' \}, j \in \{ 1,...,I+J \}, t\in \{1,...,T\}$.

\begin{figure}
	\centering
	\includegraphics[width=8cm]{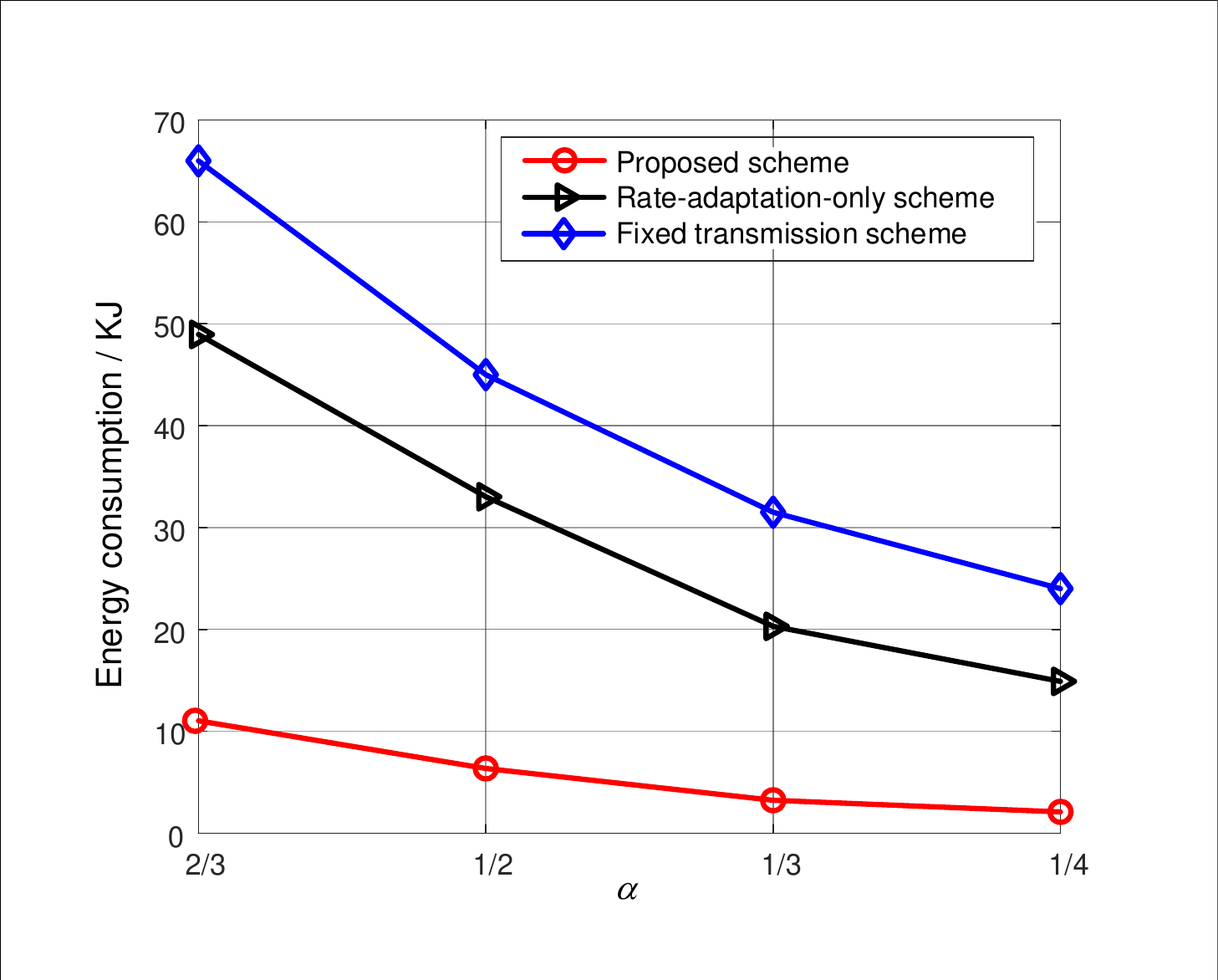}
	\caption{Energy consumption of the network under different schemes.}
	\label{fig3}
\end{figure}

Fig.~\ref{fig3} shows that  for $\alpha = 2/3$,
the network energy consumption is reduced by $83\%$ compared to the fixed transmission scheme and
$77\%$ compared to the rate-adaptation-only scheme.
The performance gains increase up to $86\%$ and $91\%$
for $\alpha = 1/4$.
It indicates that the energy consumption of the network can be prominently reduced with the proposed scheme.
Besides, the performance gain increases when $\alpha$ decreases. 
This is reasonable, as 
the smaller $\alpha$ is, the less time slots each vessel needs to occupy to satisfy its QoS guarantees.
In this case, the network has more potentials to improve the energy efficiency.

The complexity of the proposed scheme is illustrated in Fig.~\ref{fig5}
for the $10$ randomly sampled network topologies, respectively, in terms of the number of problems solved in Algorithm~\ref{Iterative_scheme}, i.e., $\mathfrak{N}$.
It can be seen that the proposed scheme has a rather low computation complexity
compared to the optimal scheme, the maximum complexity of which is $O\left( 2^{\left( I+J \right)^2T} \right)$,
i.e., $O\left( 2^{ 1000} \right)$ in the considered network.
Further, in all the $10$ randomly sampled topologies, $\mathfrak{N}$
is less than $1\%$  of $\bar{\mathfrak{N}} =  (I+J)^2(T-1)T [ (I+J)^2 + I+J -N ]$.
That is that $\bar{\mathfrak{N}} $ is a quite loose upper bound for $\mathfrak{N}$.

\begin{figure}
	\centering
	\includegraphics[width=8cm]{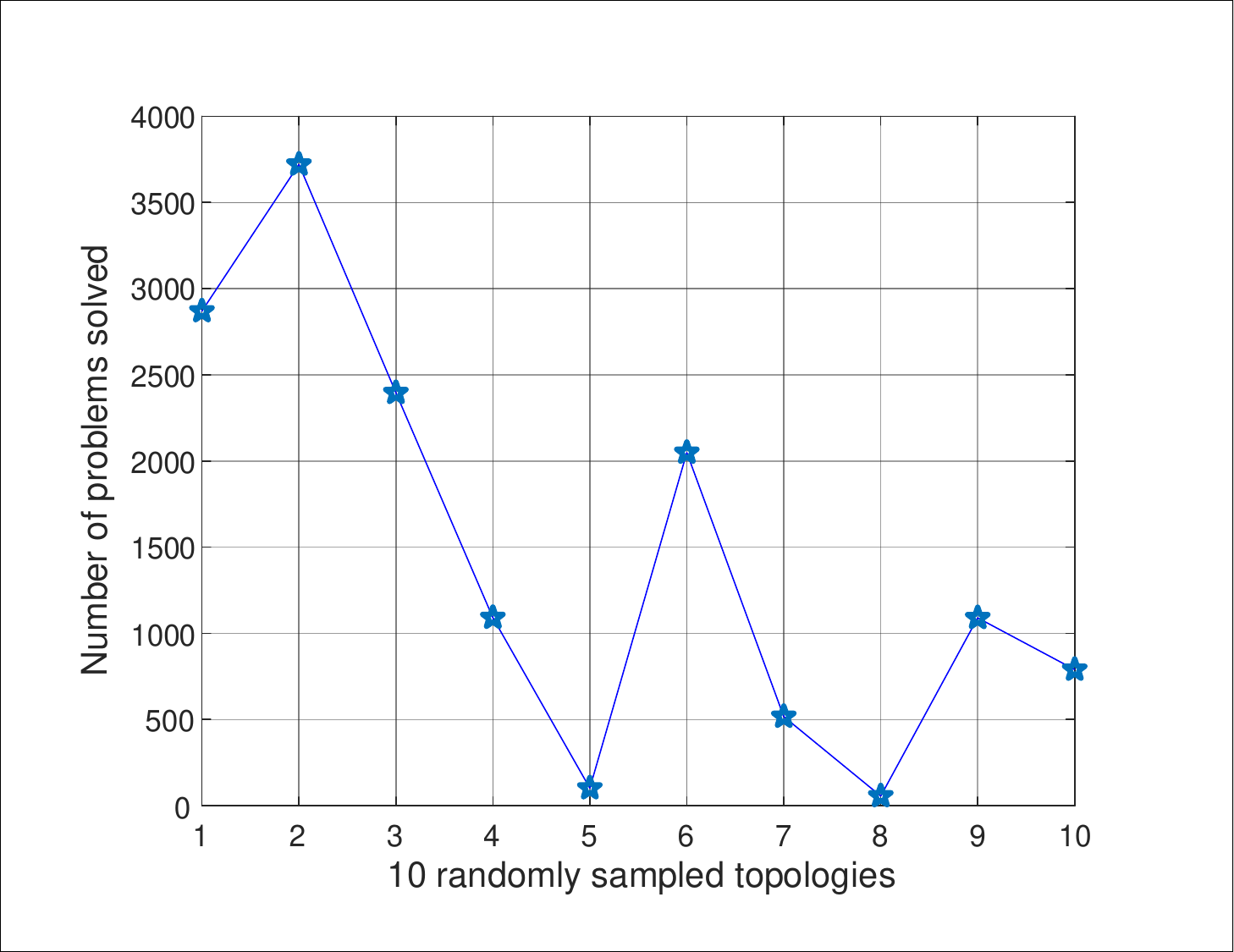}
	\caption{Total number of problems solved, i.e., $\mathfrak{N}$, for $10$ randomly sampled topologies.}
	\label{fig5}
\end{figure}

To show the convergence characteristic of the proposed scheme, 
the variation of the number of problems solved in each iteration
along with the implementation process of the scheme, for the randomly sampled topology shown in Fig.~\ref{fig2r}, 
is shown in  Fig.~\ref{fig4}.
Note that $N_s = N_{s,1}$ when $ 0 \leq s \leq \bar{s}_1-1$, and $N_s = N_{s-\bar{s}_1,2}$ when $ s \geq \bar{s}_1$,
where $\bar{s}_1$ is the value of $s_1$ after the iteration for $x=1$ is ended.
It indicates that the proposed joint link scheduling and rate adaptation scheme converges rather quickly
for different values of $\alpha$, i.e., different QoS guarantees for the vessels.

\begin{figure}
	\centering
	\includegraphics[width=8cm]{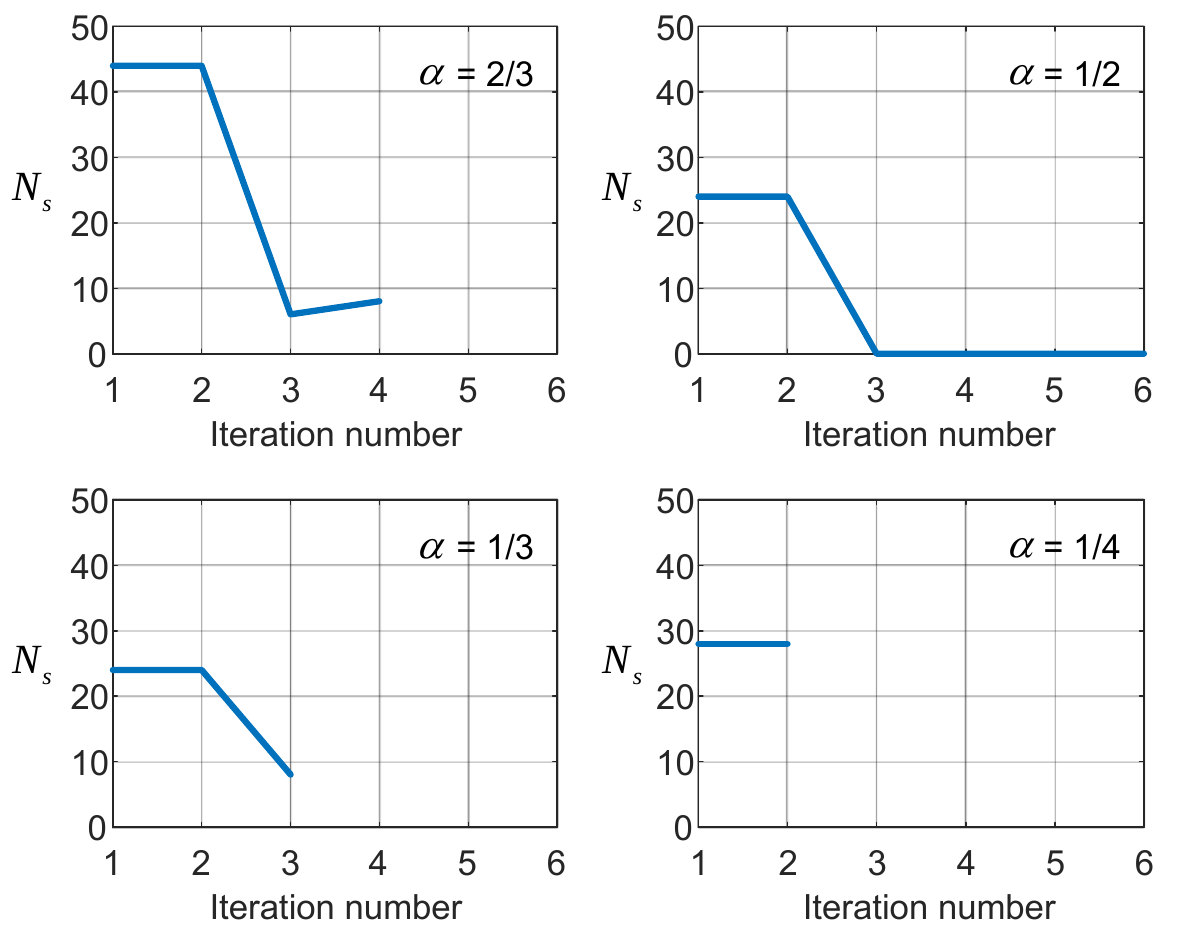}
	\caption{Variation of $N_s$ along with the iterations.}
	\label{fig4}
\end{figure}

\subsection{Comparison with the Optimal Solution}
In this subsection, we explore the performance gap between the proposed scheme and
the optimal solution,
which can be obtained via exhaustive search, 
however with much higher complexity that becomes prohibitive in practice, especially for a large network.
Considering that the implementation complexity of the optimal scheme is too high for our considered network scenario,
we turn to finding an upper bound for the optimal performance as a comparison benchmark for the proposed scheme.

Suppose $\{ r^{optml}_{i,j,t}  \}, \{ \delta^{optml}_{i,j,t}  \}$ is an optimal solution to the 
joint link scheduling and rate adaptation problem shown in~(\ref{eq_3_7}),
and $\{ r^{up}_{i,j,t}  \}$ is an optimal solution to the relaxed problem in~(\ref{eq_3_14}).
It is easy to know that 
\begin{align}{\label{eq_3_70}}
	\tilde{E}_{total} ( \{ r^{up}_{i,j,t}  \} ) \leq \tilde{E}_{total} ( \{  r^{optml}_{i,j,t} \} ).
\end{align}
Thus, we can take $\{ r^{up}_{i,j,t}  \}$, which is a solution corresponding to a super-optimal scheme, 
as the comparison benchmark in our exploring of the performance gap between our proposed scheme and the optimal solution.
Note that $\{ r^{up}_{i,j,t}  \}$ can be efficiently obtained by solving the equivalent Min-Max problem shown in~(\ref{eq_3_17}).

The energy consumptions of the network with the proposed scheme and $\{ r^{up}_{i,j,t}  \}$
averaged over the $10$ randomly sampled network topologies are compared in Fig.~\ref{fig6}.
It can be seen that the performance of the proposed scheme is quite close to that the super-optimal scheme,
with a gap of about $10\%$.
Correspondingly, it is concluded that the performance gap between the proposed scheme and the optimal solution must be even less. 
Therefore, the proposed scheme provides an effective suboptimal solution
to the joint link scheduling and rate adaptation for the hybrid MCN.

\begin{figure}[t]
	\centering
	\includegraphics[width=8cm]{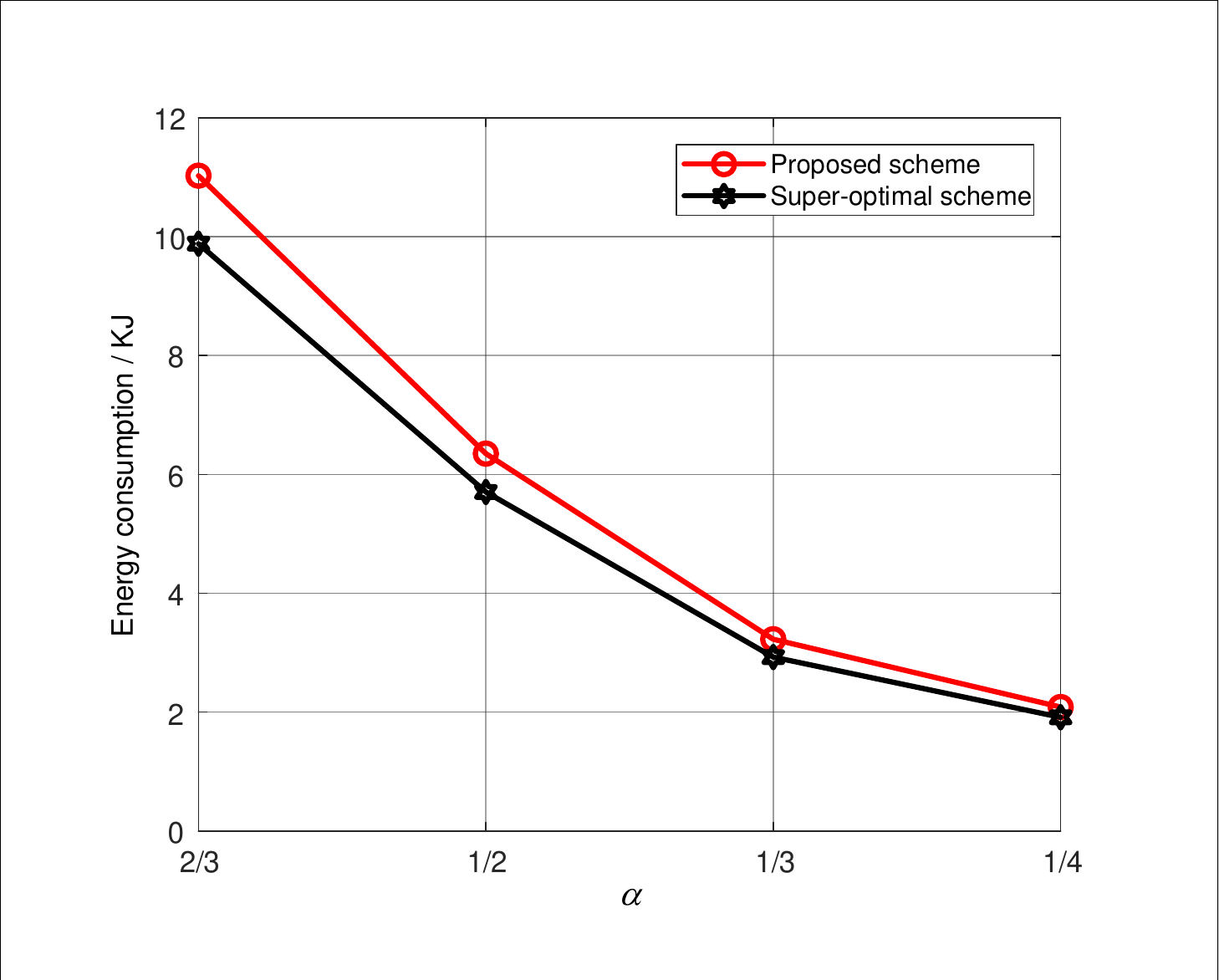}
	\caption{Illustration of the suboptimality of the proposed scheme.}
	\label{fig6}
\end{figure}

\begin{figure} [t]
	\centering
	\includegraphics[width=8cm]{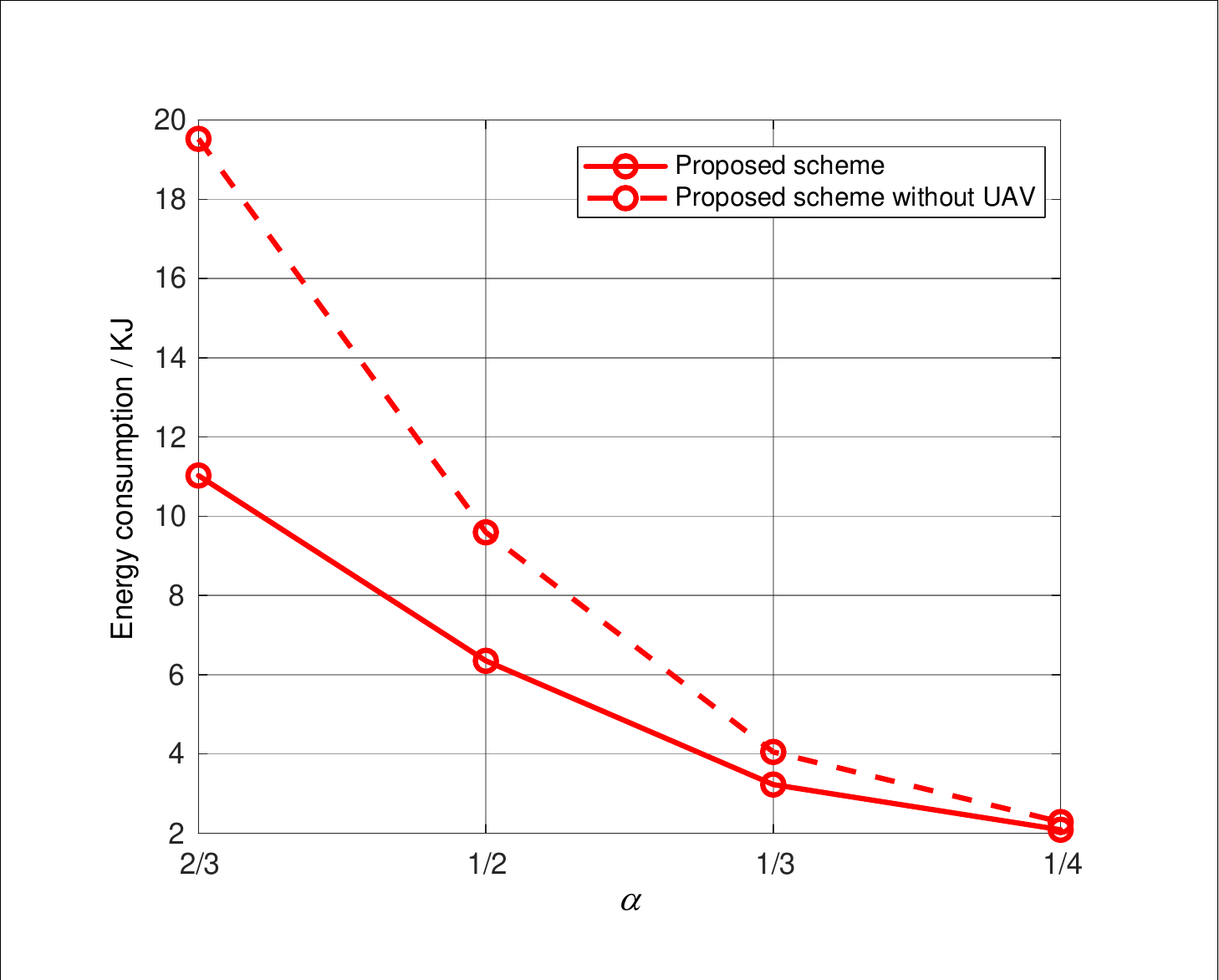}
	\caption{Illustration of the contribution of UAVs on the network energy efficiency.}
	\label{fig7}
\end{figure}

\subsection{Network Performance with and without UAVs}
The maritime UAVs play a critical role for both energy efficiency improvement and coverage enhancement in the MCN.
In this subsection, we try to explore the performance gain that can be achieved by UAVs with the proposed joint 
link scheduling and rate adaptation scheme for the hybrid MCN. 
Specifically, the energy consumption of the MCN with UAVs and that without UAVs are firstly compared,
for the illustration of the performance gain in energy efficiency.
Then, 
the minimum data volume received by the vessels within the service duration, 
i.e., $\mathbb{V}_{min} = \min_{j \in \{ I+1,I+J \}} V_{j,T} $, 
is compared for two different network configurations based on the proposed algorithm,
so as to demonstrate the potentials of UAVs with the proposed scheme in coverage enhancement.

For the energy consumption comparison,
the simulations are conducted under the same $10$ random network topologies described previously.
 As shown by Fig.~\ref{fig7},
the energy consumption of the UAV-aided network is reduced by $78\%$ for $\alpha = 2/3$ and $10\%$ for $\alpha=1/4$,
as compared to that without UAV relaying the data.
It indicates that the UAV is quite helpful for the reduction of the network energy consumption,
especially when large data volumes are required by the vessels.

	As to the demonstration of the potentials in coverage enhancement, 	
	the two network configurations considered are designated as C1 and C2, respectively.
	C1 refers to the UAV-aided MCN with the proposed scheme,
	and C2 represents the MCN in which the fixed transmission scheme illustrated in Subsection A is implemented.
	Note that as the on-shore BS transmits to the vessels directly for the fixed transmission scheme, the UAV just keeps idle in C2.
	The simulations are carried out under a group of $10$ semi-random network topologies with some ``coverage hole".
	Specifically, the shipping lanes of the vessels corresponding to $j=7, 8, 9, 10$ are restricted within the subarea where
	the transmission distance from the on-shore BS is larger than $5$km to form a ``coverage hole",
	while that of the other vessels are still randomly-generated across the whole square area.
	Besides, a semi-random trajectory is generated for the UAV in each topology, 
	the starting point of which is restricted within the subarea with a 
	transmission distance less than $3$km with respect to the BS, and the ending point
	is randomly generated in the same subarea where the shipping lanes of the vessels with $j=7, 8, 9, 10$ are located, 
	with an expectation that the UAV can help to reduce the ``coverage hole".

\begin{figure} [t]
	\centering
	\includegraphics[width=8cm]{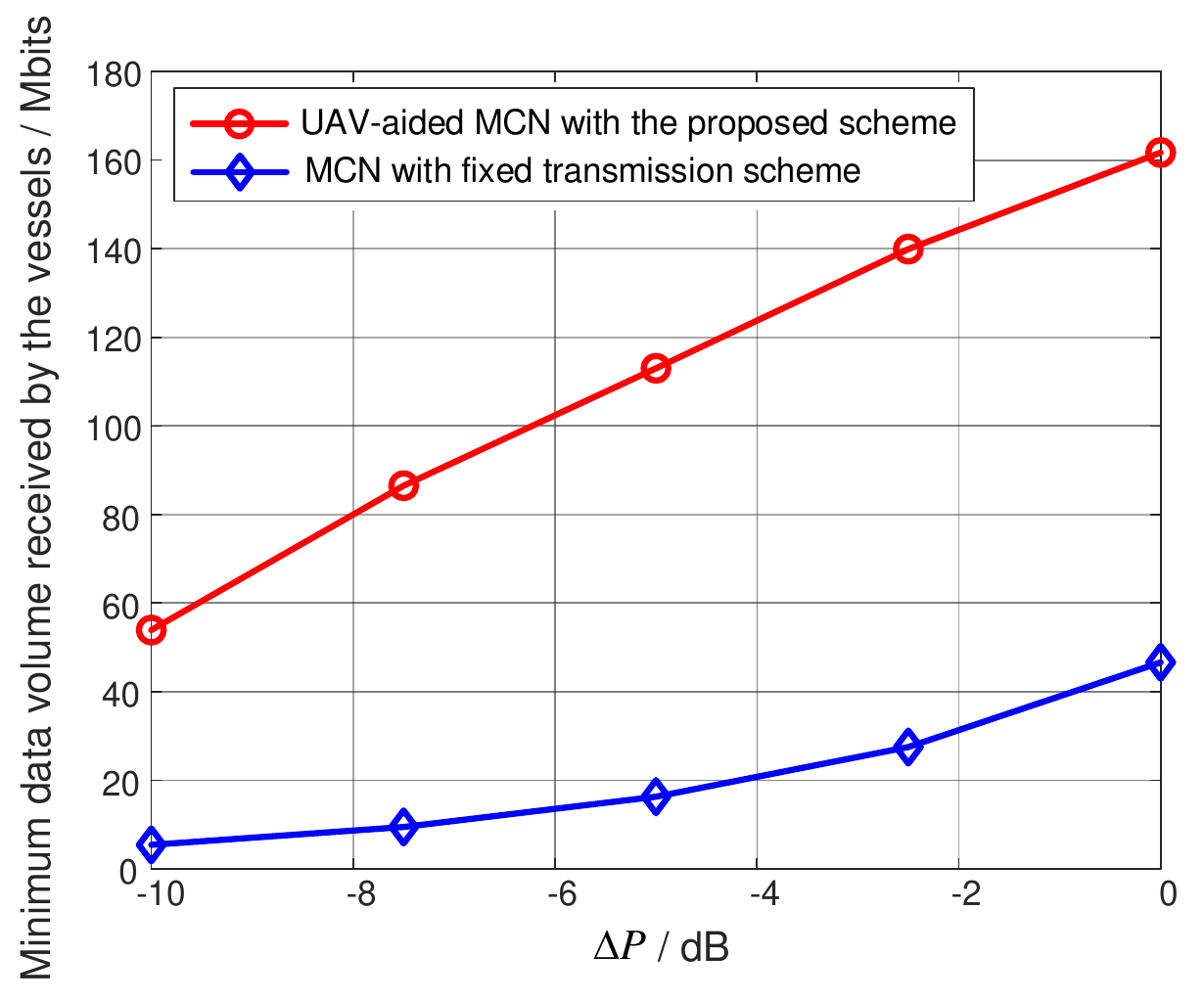}
	\caption{Illustration of the potentials of UAVs with the proposed scheme in coverage enhancement.}
	\label{fig8}
\end{figure}

	The minimum data volumes received by the vessels under C1 and C2, denoted as $\mathbb{V}_{min}^{(1)}$ and $\mathbb{V}_{min}^{(2)}$ respectively,
	are presented in Fig.~\ref{fig8}. The results are obtained by an average of $10$ semi-random network topologies.
	Different transmit powers for the on-shore BS as well as the UAV and the vessels are considered to enrich the comparison.
	Specifically, the transmit power of the on-shore BS and that of the UAV and the vessels
	are respectively set as $P'_0 = P_0 (\text{dB}) + \Delta P$ and $P'_i = P_i (\text{dB}) + \Delta P$, $1 \leq i \leq 9$,
	with $P_0 =50$W, $P_i = 10$W, $1 \leq i \leq 9$, and $\Delta P$ taking value from $-10$dB to $0$dB. 
	Note that one  primary indication for a better coverage is that a larger $\mathbb{V}_{min}$ 
	could be achieved.
	To show the potentials of UAVs with the proposed scheme in coverage enhancement, 
	$\mathbb{V}_{min}^{(1)}$ is obtained 
	based on Algorithm~\ref{Iterative_scheme} via a bisection searching within $[\mathcal{V}_b, ~\alpha_b \mathcal{V}_b]$,
	where $ \mathcal{V}_b = \min_{j} \sum_{t=1}^T R_{0,j,t} \Delta \tau$.
	Specifically, the maximum feasible $V_j^{QoS} \in [\mathcal{V}_b, ~\alpha_b \mathcal{V}_b]$, $j=I+1,...,I+J$, 
	with $t_j^{QoS} = T$, $j=I+1,...,I+J$,
	is searched based on C1 for each of the $10$ topologies,
	and its average with respect to the network topologies
	is shown in Fig.~\ref{fig8} as a benchmark value for $\mathbb{V}_{min}^{(1)}$ in the comparison.
	With a moderate complexity and without loss of generality, $\alpha_b = 10$ is adopted in the simulations.\footnote{Note that as an upper bound $10$ is set for $\alpha_b$, $\mathbb{V}_{min}^{(1)} = 10\mathbb{V}_{min}^{(2)}$ in Fig.~\ref{fig8} may indicate that $\mathbb{V}_{min}^{(1)} > 10\mathbb{V}_{min}^{(2)}$.}
	While for $\mathbb{V}_{min}^{(2)}$, it is simply set as the average of $\mathcal{V}_b$ in the $10$ topologies.  
	It can be seen from Fig.~\ref{fig8} that $\mathbb{V}_{min}^{(1)}$ is much larger than $\mathbb{V}_{min}^{(2)}$,
	i.e., C1 could achieve a noticeable gain in coverage enhancement compared to C2, just as expected.
	Furthermore, while the average performance gain of C1 over C2, i.e., $\mathbb{V}_{min}^{(1)}/\mathbb{V}_{min}^{(2)}$, 
	is about $3.5$ for $\Delta P = 0$dB, it is $10$ for $\Delta P = -10$dB.
	It indicates that maritime UAVs with the proposed scheme have greater potentials in coverage enhancement for MCNs 
	with more severe coverage problems.

\section{Conclusions}
In this paper, we focused on the coverage extension problem of 5G on the ocean. For higher benefit-to-cost ratio, we have proposed to orchestrate 
hierarchical links for a hybrid satellite-UAV-terrestrial
MCN, using only the slowly-varying large-scale CSI, which can be obtained easily according to the trajectories of UAVs and the shipping lanes of vessels.  
A MINLP problem has been formulated for the minimization of the network energy consumption
with QoS guarantees for all users. 
To tackle the NP-hard problem,
a suboptimal scheme has been proposed based on a relaxation and gradually-approaching method following the gentlest-descent principle,
as well as an equivalent Min-Max transformation. 
Simulations results have shown that the proposed process-oriented joint link scheduling and rate adaptation scheme converges quickly, and has a polynomial computation complexity.
Furthermore, it can achieve a prominent performance close to the optimal solution, 
which, however, cannot be easily implemented in practice due to its prohibitive complexity.
By adopting the proposed scheme, all users can share 
an agile on-demand coverage, and the total energy consumption can be greatly reduced, making it a promising technique to be applied in future greener and more affordable MCNs using both fixed and
deployable infrastructures. Based on the results of this work, more systematic designs incorporating with e.g., UAV trajectory optimization, online dynamic adjustment, and more critical QoS requirements, could be uncovered in future works.

\appendices
\section{Proof of Theorem~\ref{theorem_min}}
The first-order and the second-order partial derivative of $f\left(  \{r_{i,j,t} \},  \{z_{i,j,t} \} \right)$ with respect to $z_{i,j,t}, \forall i,j,t$, can be respectively derived as
\begin{align}\label{eq_eA_1}
	&\frac{ \partial f\left(  \{r_{i,j,t} \},  \{z_{i,j,t} \} \right) }{ \partial z_{i,j,t} }  \nonumber \\
	& = \frac{\sigma^2 \Delta \tau}{\beta_{i,j,t}} \left(  \frac{ 2^{ \frac{1}{B_s} r_{i,j,t} + \log_2e(1-z_{i,j,t}^{-1})}}{z_{i,j,t}^2} - 1  \right),
\end{align} 
\begin{align}\label{eq_eA_2}
	&\frac{ \partial^2 f\left(  \{r_{i,j,t} \},  \{z_{i,j,t} \} \right) }{ \partial^2 z_{i,j,t} }  \nonumber \\
	& = \frac{\sigma^2 \Delta \tau 2^{\frac{1}{B_s} r_{i,j,t} + \log_2e(1-z_{i,j,t}^{-1})} }{\beta_{i,j,t} z_{i,j,t}^3 } \left(  \frac{ 1}{z_{i,j,t}} - 2  \right).
\end{align} 

By substituting the expression of $r_{i,j,t}$ in (\ref{eq_3_4}) into $2^{ \frac{1}{B_s} r_{i,j,t} + \log_2e(1-z_{i,j,t}^{-1})}$ 
and setting $z_{i,j,t} = W_{i,j,t}$, we get
\begin{align}\label{eq_eA_3}
	&~ 2^{ \frac{1}{B_s} r_{i,j,t} + \log_2e(1-W_{i,j,t}^{-1})} \nonumber \\
	&= 2^{ \log_2 \left( 1+ \frac{ \beta _{i,j,t } p_{i,j,t} W_{i,j,t}^{-1} }{ \sigma ^2 }  \right) + \log_2(W_{i,j,t}) } \nonumber \\
	&= \frac{ \sigma ^2 + \beta _{i,j,t } p_{i,j,t} W_{i,j,t}^{-1} }{ \sigma ^2 }  W_{i,j,t}.
\end{align}
Based on (\ref{eq_3_5}), it is known that
\begin{eqnarray}\label{eq_eA_4}
	1 = W_{i,j,t}^{-1} + \frac{ \beta _{i,j,t } p_{i,j,t} W_{i,j,t}^{-1} }{ \sigma ^2 + \beta _{i,j,t } p_{i,j,t} W_{i,j,t}^{-1} }.
\end{eqnarray} 
(\ref{eq_eA_4}) indicates that 
\begin{eqnarray}\label{eq_eA_5}
	W_{i,j,t}^{-1} = \frac{ \sigma ^2 }{ \sigma ^2 + \beta _{i,j,t } p_{i,j,t} W_{i,j,t}^{-1} },
\end{eqnarray} 
and further
\begin{eqnarray}\label{eq_eA_6}
	W_{i,j,t} = \frac{ \sigma ^2 + \beta _{i,j,t } p_{i,j,t} W_{i,j,t}^{-1} }{\sigma ^2  }.
\end{eqnarray} 
Substitute (\ref{eq_eA_6}) into (\ref{eq_eA_3}), and it is obtained that
\begin{eqnarray}\label{eq_eA_7}
	2^{ \frac{1}{B_s} r_{i,j,t} + \log_2e(1-W_{i,j,t}^{-1})} = W_{i,j,t}^2.
\end{eqnarray} 

It can be seen from (\ref{eq_eA_1}) and (\ref{eq_eA_7}) that
\begin{eqnarray}\label{eq_eA_8} 
	\left. \frac{ \partial f\left(  \{r_{i,j,t} \},  \{z_{i,j,t} \} \right) }{ \partial z_{i,j,t} } \right|_{z_{i,j,t}= W_{i,j,t}} = 0.
\end{eqnarray} 
Furthermore, (\ref{eq_eA_2}) implies 
\begin{eqnarray}\label{eq_eA_9}
	\left. \frac{ \partial^2 f\left(  \{r_{i,j,t} \},  \{z_{i,j,t} \} \right) }{ \partial^2 z_{i,j,t} } \right|_{z_{i,j,t} \geq 1} < 0,
\end{eqnarray} 
which shows that $\frac{ \partial f\left(  \{r_{i,j,t} \},  \{z_{i,j,t} \} \right) }{ \partial z_{i,j,t} }$
is decreasing with respect to $z_{i,j,t}$ within $[ 1, +\infty)$.
Thus, 
\begin{eqnarray}\label{eq_eA_10} 
	\left. \frac{ \partial f\left(  \{r_{i,j,t} \},  \{z_{i,j,t} \} \right) }{ \partial z_{i,j,t} } \right|_{z_{i,j,t} < W_{i,j,t}} > 0,
\end{eqnarray} 
\begin{eqnarray}\label{eq_eA_11} 
	\left. \frac{ \partial f\left(  \{r_{i,j,t} \},  \{z_{i,j,t} \} \right) }{ \partial z_{i,j,t} } \right|_{z_{i,j,t} > W_{i,j,t}} < 0.
\end{eqnarray} 
Based on (19), (35), (\ref{eq_eA_8}), (\ref{eq_eA_10}), and (\ref{eq_eA_11}), it can be inferred that
\begin{eqnarray}\label{eq_eA_12} 
	\tilde{E}_{\rm{total}} \left( \{r_{i,j,t} \} \right) = \max_{z_{i,j,t} \geq 1, \forall i,j,t }  \,\, f\left(  \{r_{i,j,t} \},  \{z_{i,j,t} \} \right).
\end{eqnarray} 

Based on the convexity of the exponential function, it is easy to know that
$f\left(  \{r_{i,j,t} \},  \{z_{i,j,t} \} \right)$ is convex with respect to $\{r_{i,j,t}\}$ for $r_{i,j,t} \geq 0, \forall i,j,t$.
Besides, based on (\ref{eq_eA_9}), we know that
$f\left(  \{r_{i,j,t} \},  \{z_{i,j,t} \} \right)$ is concave with respect to $\{ z_{i,j,t} \}$ for $z_{i,j,t} \geq 1, \forall i,j,t$.

\section{Proof of Theorem~\ref{theorem_converge}}
Based on the illustration in Algorithm~\ref{Iterative_scheme},
it can be seen that the solution space for the relaxed problem~(\ref{eq_3_14}) is shrinked 
in each iteration $s_x$, $x=1,2$,
based on the process-oriented rule defined by~(\ref{eq_3_55})
and (\ref{eq_30})-(\ref{eq_34}).

It can be inferred that for $x=1$, $1 \leq \mathfrak{e}_{s_1,1} \leq T $ half-duplex mode constraints given in~(\ref{eq_3_7}b)
are assured to be transformed from the violated state to the satisfied state after each iteration $s_1$. 
To obtain the maximum iteration number, we consider the worst case, where
the constraints given in~(\ref{eq_3_7}b) are all violated
with the initial $\tilde{r}_{i,j,t}$ and $\tilde{\delta}_{i,j,t}$
obtained by solving the relaxed problem~(\ref{eq_3_14}).
Denote the number of iterations needed for $x=1$ in the worst case as $\bar{s}_1$,
and it should satisfy
\begin{eqnarray}\label{eq_eB_1} 
\sum_{s_1=0}^{\bar{s}_1-1} \mathfrak{e}_{s_1,1} \geq (I+J)T.
\end{eqnarray}
Thus, 
\begin{eqnarray}\label{eq_eB_2} 
\bar{s}_1 \leq (I+J)T.
\end{eqnarray} 
Furthermore, we have
\begin{eqnarray}\label{eq_eB_3} 
s_1 \leq \bar{s}_1 \leq  (I+J)T.
\end{eqnarray} 

After the iterations for $x=1$ terminate, at most $I+J$ links are scheduled to be active in time slot $t \in \{1,...,T\}$.
It indicates that at most $I+J-N$ extra active links should be removed, i.e., set idle, for each time slot $t$
so as to satisfy the subcarrier number constraints given in~(\ref{eq_3_7}c). 
That is that at most $(I+J-N)T$ extra active links in total should be set idle through the iterations for $x=2$.
With each iteration $s_2$, $ 1\leq \mathfrak{e}_{s_2,2} \leq T$ extra active links could be set idle. 
Thus, if we denote the number of iterations needed for $x=2$ in the worst case as $\bar{s}_2$, then
\begin{eqnarray}\label{eq_eB_4} 
\bar{s}_2 \leq  (I+J-N)T.
\end{eqnarray} 
Furthermore, we have 
\begin{eqnarray}\label{eq_eB_3} 
s_2 \leq \bar{s}_2 \leq  (I+J-N)T.
\end{eqnarray} 

Finally, based on~(\ref{eq_eB_2}) and~(\ref{eq_eB_4}), 
we can obtain the total number of iterations needed for the proposed joint link scheduling and rate adaptation scheme
in Algorithm~\ref{Iterative_scheme}, i.e., 
\begin{eqnarray}\label{eq_eB_5} 
s \leq \bar{s}_1 + \bar{s}_2 \leq  [2(I+J)-N]T.
\end{eqnarray}

\end{document}